\documentclass[pra,superscriptaddress,twocolumn]{revtex4}

\usepackage{amsfonts,amssymb,amsmath,amsthm}
\usepackage{graphics,graphicx,epsfig}
\usepackage{tikz}
\usepackage{hyperref}

\usetikzlibrary{calc}

\def\identity{\leavevmode\hbox{\small1\kern-3.8pt\normalsize1}}

\newtheorem{lemma}{Lemma}

\newtheorem{prop}{Proposition}

\newcommand{\Ab}[1]{ \left| #1 \, \right|} 
\newcommand{\ket}[1]{\left | #1 \right\rangle}
\newcommand{\bra}[1]{\left \langle #1 \right |}

\newcommand{\Tr}{\mathrm{Tr}}

\renewcommand{\epsilon}{\varepsilon}
\newcommand\bigforall{\mbox{\Large $\mathsurround=0pt\forall$}}

\begin{document}

\title{Trade-offs in multi-party Bell inequality violations in qubit networks}
\author{Ravishankar \surname{Ramanathan}}
\affiliation{Laboratoire d'Information Quantique, Universit\'{e} Libre de Bruxelles, Belgium}
\author{Piotr Mironowicz}
\affiliation{Department of Algorithms and System Modeling, Faculty of Electronics, Telecommunications and Informatics, Gda\'nsk University of Technology}
\affiliation{National Quantum Information Centre in Gda\'nsk, 81-824 Sopot, Poland}

\date{\today}

\begin{abstract}
	Two overlapping bipartite binary Bell inequalities cannot be simultaneously violated as this would contradict the usual no-signalling principle.
	This property is known as monogamy of Bell inequality violations and generally Bell monogamy relations refer to trade-offs between simultaneous violations of multiple inequalities.
	It turns out that multipartite Bell inequalities admit weaker forms of monogamies that allow for violations of a few inequalities at once.
	Here we systematically study monogamy relations between correlation Bell inequalities both within quantum theory and under the sole assumption of no signalling.
	We first investigate the trade-offs in Bell violations arising from the uncertainty relation for complementary binary observables, and exhibit several network configurations in which a tight trade-off arises in this fashion. 
	We then derive a tight trade-off relation which cannot be obtained from the uncertainty relation showing that it does not capture monogamy entirely.
	The results are extended to Bell inequalities involving different number of parties and find applications in device-independent secret sharing and device-independent randomness extraction.
	Although two multipartite Bell inequalities may be violated simultaneously, we show that genuine multi-party non-locality, as evidenced by a generalised Svetlichny inequality, does exhibit monogamy property.
	Finally, using the relations derived we reveal the existence of flat regions in the set of quantum correlations.

\end{abstract}

\maketitle

\textit{Introduction.-}
Measurements on spatially separated entangled systems lead to correlations that do not conform to local realism, as evidenced by the violation of Bell inequalities~\cite{Brunner14}. This non-locality of quantum systems is of great interest, not only for fundamental reasons but also as a resource in applications such as device-independent randomness generation~\cite{Acin16} and key distribution~\cite{Gisin02,Scarani09}, in reductions of communication complexity~\cite{Buhrman10}, etc.  In the applications of quantum non-local correlations to cryptography, a key role is played by the phenomenon of ``monogamy" of non-local correlations~\cite{Coffman00}. 

Any two quantum systems that exhibit maximally non-locality with each other, cannot exhibit non-local correlations (or even classical correlations) with any other third system. Therefore, non-locality is a resource that cannot be freely shared among different parties.
Tight quantitative trade-off relations for Bell inequality violations have been shown in some cases such as the well-known CHSH inequality~\cite{TV06}. Moreover, the limited shareability property of correlations has even been shown to extend beyond quantum theory to general no-signaling theories (the class of theories which do not allow for arbitrarily fast signaling)~\cite{Toner09}. In this context, it is pertinent to point out recent results that show that derivation of monogamy relations from relativistic causality alone depends on the spacetime configuration of the measurement parties~\cite{HR16}. 

While the shareability of two-party non-locality is well studied, much less is known about the trade-off relations in the case of multi-party non-locality. A preliminary study was carried out in~\cite{KPRLK11} where an uncertainty-type relation for dichotomic quantum observables termed correlation complementarity was shown to underlie many multi-party monogamy relations for correlation Bell inequalities involving two dichotomic observables per party. This question has gained importance with the development of cryptographic protocols involving many parties, as well as with the substantial experimental progress in the engineering and detection of such correlations~\cite{EMK+14,MSB+11,AMJ+12}. 


In consideration of the intrinsic relevance of the notion of monogamy to the foundational core of quantum correlations, it has become a worthwhile objective to deeply explore the characteristics of multipartite non-locality distributed over many parties, and to establish what usefulness to multi-user quantum communication protocols can such a resource provide. 
In this paper, we carry out a detailed study of trade-offs in the violation of such multi-party correlation inequalities in qubit networks (as we shall see by the well-known Jordan lemma~\cite{Jordan1875} no loss of generality in restricting to qubits). 
We first investigate following~\cite{KPRLK11} the trade-offs in such violations arising from an uncertainty relation for complementary binary observables, and exhibit novel constructions of network configurations in which a tight trade-off arises from this complementarity. We then go beyond the analysis in~\cite{KPRLK11} and show that the uncertainty relation does not capture monogamy entirely by deriving a tight trade-off relation in a ladder network consisting of an arbitrary odd number of qubits. We extend the considerations to deriving trade-offs between Bell violations for different number of parties, with potential applications for device-independent secret sharing. We apply the derived trade-off relations to bound the guessing probability of any party's outcome by an adversary, a quantity of central importance in device-independent randomness extraction~\cite{AMP12, PAM+10}. We consider a generalization of the well-known Svetlichny inequality~\cite{Sve87} that witnesses genuine multi-party non-locality to an arbitrary number of inputs. We then use it to show that while multi-party Bell inequality violation is by itself not monogamous, the notion of genuine multi-party non-locality as evidenced by the well-known Svetlichny inequality does exhibit monogamy, any violation beyond the threshold required to witness genuine multi-party non-locality by a subset of the parties precludes its violation by any other subset of the parties with non-zero overlap with the first set. The trade-off relations also give rise to Tsirelson bounds on a class of inequalities with few-body correlators, we show that these inequalities can be violated by multiple distinct quantum boxes detecting certain flat regions in the set of quantum correlations. Finally, we investigate the no-signaling analogues of the tight hypersphere monogamy relations within quantum theory, and derive a general linear no-signaling trade-off relation for arbitrary Bell inequalities extending the analysis for two-party inequalities in~\cite{PB09}.

\textit{Preliminaries.-}
The $(n,m,d)$ Bell scenario denotes the situation where $n$ parties choose from among $m$ measurements each obtaining one of $d$ outcomes. While one may also consider scenario with a different number of inputs and outputs per party, this will not concern us in this paper. The phenomenon of monogamy of violation of the famous Clauser-Horne-Shimony-Holt (CHSH)~\cite{CHSH} inequality was quantitatively shown by Toner in~\cite{Toner09} in general no-signaling theories, while the exact trade-off in its violation within quantum theory was shown by Toner and Verstraete in~\cite{TV06}. This manifested in the $(3,2,2)$ Bell scenario where three parties Alice, Bob and Charlie each measure one of two binary observables $\{A_1, A_2\}$, $\{B_1, B_2\}$ and $\{C_1, C_2\}$ respectively. Quantitatively,  the correlations between their measurement results obey
\begin{eqnarray}
\label{eq:CHSH-mono}
\langle \text{CHSH} \rangle_{\text{AB}} + \langle \text{CHSH} \rangle_{\text{AC}} & \leq & 4 \; \; \; \; \text{in gen. NS theories} \nonumber \\
\langle \text{CHSH} \rangle_{\text{AB}}^2 + \langle \text{CHSH} \rangle_{\text{AC}}^2 & \leq & 8 \; \; \; \; \text{in quantum theory} \nonumber \\
\end{eqnarray}
for the CHSH-Bell expression $\langle \text{CHSH} \rangle_{\text{AB}} := \langle A_1 B_1 + A_2 B_1 + A_1 B_2 - A_2 B_2 \rangle \leq 2$.


Here, we generalize this result to the correlation based Bell inequalities in the $(n,2,2)$ scenario (also known as multi-party binary XOR games). In the $(n,2,2)$ Bell scenario, the $i$-th party measures the binary observables $A^{(i)}_{x_i}$ with $i \in [n]$ and $x_i \in \{1,2\}$. The entire class of $(n,2,2)$ inequalities involving $n$-body correlation functions $E(x_1, \dots, x_n) = \langle \otimes_{i=1}^{n} A^{(i)}_{x_i} \rangle$ was derived by Werner and Wolf in~\cite{WernerWolf01} and independently by \.{Z}ukowski and Brukner in~\cite{ZukowskiBrukner02}. 
These inequalities define the facets of a polytope which is isomorphic to the hyper-octahedron (also known as the cross-polytope). The $2^{2^n}$ facet-defining inequalities can all be described by a single inequality given as
\begin{equation}
\label{eq:gen-bell}
\sum_{s_1, \dots, s_n = -1,1} \big\lvert \sum_{x_1, \dots, x_n = 1,2} s_1^{x_1 - 1} \dots s_n^{x_n - 1} E(x_1, \dots, x_n) \big\rvert \leq 2^n. 
\end{equation}
Note that due to the isomorphism with the cross-polytope, these facets are all simplices. The violation of the general Bell inequality (\ref{eq:gen-bell}) by an $n$-qudit state has also been widely studied. 

A well-known result known as the Jordan Lemma~\cite{Jordan1875} states that any pair of Hermitian unitaries admit a common block diagonalisation in blocks of dimension no more than $2$. We can thus set
$A^{(i)}_{x_i} = \oplus_{k} (A^{(i)}_{x_i})^k$,
where the observables $(A^{(i)}_{x_i})^k$ for all $k$ are still binary Hermitian and of dimension at most $2$. This reduces the problem of finding the optimal quantum strategy to considering qubit subspaces on each party. We therefore assume that each party possesses a qubit and measures the binary observable $A^{(j)}_{x_j} = \vec{n}^{(j)}_{x_j} \cdot \vec{\sigma}$ on it. An arbitrary mixed state of $n$ qubits is given as
\begin{equation}
\rho = \frac{1}{2^n} \sum_{k_1, \dots, k_n = 0}^{3} T_{k_1, \dots, k_n} \sigma^{(1)}_{k_1} \otimes \dots \otimes \sigma^{(n)}_{k_n}
\end{equation} 
where $\sigma^{(j)}_{k_j}$ are the usual Pauli operators $\{\mathbf{1}, \sigma_x, \sigma_y,  \sigma_z \}$, $\vec{\sigma} = (\sigma_x, \sigma_y, \sigma_z)$ and the real coefficients $T_{k_1, \dots, k_n}$ form the correlation tensor $\hat{T}$. The two measurement directions for each party $\vec{n}^{(j)}_1$ and $\vec{n}^{(j)}_2$ span a plane and can be described by the orthogonal measurement directions $\vec{o}^{(j)}_1 \perp \vec{o}^{(j)}_2$ by
\begin{eqnarray}
\vec{n}^{(j)}_1 &= & \vec{o}^{(j)}_1 \cos{\left( \theta_j + \frac{\pi}{2} \right)} -  \vec{o}^{(j)}_2 \cos{\left(\theta_j + \pi \right)} \nonumber \\
\vec{n}^{(j)}_2 &= & \vec{o}^{(j)}_1 \cos{\left(\theta_j + \frac{\pi}{2}\right)} +  \vec{o}^{(j)}_2 \cos{\left(\theta_j + \pi \right)} \nonumber \\
\end{eqnarray}
We thus arrive at the following lemma derived originally in~\cite{ZukowskiBrukner02} describing a necessary and sufficient condition for the existence of a local model for the $n$-body correlation functions of any $n$-qubit state, as well as a sufficient condition obtained from it by the application of the Cauchy-Schwarz inequality. 
\begin{lemma}[\cite{ZukowskiBrukner02}]
	\label{lem:Zukowski-Brukner}
	An $n$-qubit state $\rho$ with correlation tensor $\hat{T}$ satisfies the general correlation Bell inequality (\ref{eq:gen-bell}) if and only if for any set of local measurement planes $\text{span}(\vec{o}^{(j)}_1, \vec{o}^{(j)}_2)$ and measurement angles $\{\theta_j\}$ there holds
\begin{equation}
\sum_{k_1, \dots, k_n = 1,2} \vert T_{k_1, \dots, k_n} \vert \prod_{j=1}^{n} \cos{\left(\theta_j + k_j \frac{\pi}{2}\right)} \leq 1.
\end{equation}
A sufficient condition for the local realistic description of the $n$-qubit correlation function is given by 
\begin{equation}
\label{eq:gen-bell-2}
\sum_{k_1, \dots, k_n = 1,2} T^2_{k_1, \dots, k_n} \leq 1.
\end{equation}
\end{lemma}  
Note that the local bound has been normalized to unity for all the $(n,2,2)$ correlation inequalities. Lemma~\ref{lem:Zukowski-Brukner} is the generalization of the well-known Horodecki criterion~\cite{Horodecki95} relating the quantum violation of the CHSH inequality to the correlation tensor as 
\begin{equation}
\label{eq:Horodec-crit}
\langle \text{CHSH} \rangle^2_{AB} \leq \sum_{k_A, k_B = 1,2} 4 T^2_{k_A, k_B}. 
\end{equation}
Therefore as far as the quantum violation is concerned, it suffices to consider the sum of squares of the correlation tensor elements of a general $n$-qubit state in a plane which we take to be the $x-z$ plane without loss of generality.

 \textit{Trade-offs in qubit networks from uncertainty relations for complementary observables.-} 
 We are interested in deriving trade-off relations for the violation of the general multipartite correlation Bell inequality (binary XOR game) given in Eq.(\ref{eq:gen-bell-2}). In particular, we derive trade-off relations for the $n$-party Bell expression $\mathcal{I}^2_{l_1, \dots, l_n} = \sum_{k_{l_1}, \dots, k_{l_n} = 1,2} T^2_{k_{l_1}, \dots, k_{l_n}}$ played by the players labeled $l_1, \dots, l_n$, the local bound of the Bell expression being $1$. To simplify notation, we will denote the players by a qubit network represented by a hypergraph $H$. Each edge of the hypergraph $e \in E(H)$ will denote the subset $e = (l_1, \dots, l_n)$ of players taking part in a single game, and the value of the game achieved by the players will be denoted by $\mathcal{I}_{e}$. The trade-off relations we derive will be of the form
 \begin{equation}
 \label{eq:gen-monogamy-relation}
 \sum_{e \in E(H)} \langle \mathcal{I} \rangle^{2}_{e} \leq \vert E(H) \vert,
 \end{equation}  
 and we will refer to them as ``hyperspherical trade-off relations" in what follows.
 As a first method to derive relations of this form following~\cite{KPRLK11}, we employ the following uncertainty-type relation for complementary observables 
 \begin{lemma}[\cite{KPRLK11, WW08}]
 	\label{lem:corr-comp}
 	Let $A_1, \dots, A_m$ be binary Hermitian observables satisfying $\{A_i, A_j\} = 2 \delta_{i,j} \mathbf{1}$. Then for any quantum state $\rho$ it holds that 
 	\begin{equation}
 	\sum_{i=1}^{m} \Tr[A_i \rho]^2 \leq 1.
 	\end{equation}
 \end{lemma}
 In particular, we employ Lemma~\ref{lem:corr-comp} with observables $\otimes_{l=1}^{m} O^{(l)}_{k_{l}}$ for $k_{l} = 1, 2$ with each $O^{(l)}_{k_l} \in \{\sigma_x, \sigma_z\}$, i.e., observables corresponding to the correlation tensor element $T_{k_{1}, \dots, k_{n}}$ from the Bell expression (\ref{eq:gen-bell-2}).

 Let us illustrate the method by rederiving with it the CHSH monogamy relation of Toner and Verstraete~\cite{TV06}. The CHSH monogamy relation corresponds to the first graph in Fig.~\ref{fig:example-networks} and is given from Eqs.(\ref{eq:CHSH-mono}) and (\ref{eq:Horodec-crit}) as 
\begin{equation}
\label{eq:chsh-mono-corr}
\sum_{k_A, k_B = 1,2} T^2_{k_A, k_B} + \sum_{k_A, k_C = 1,2} T^2_{k_A, k_C} \leq 2.
\end{equation}
 The eight correlation tensor elements can be grouped into two sets of mutually anti-commuting elements as $\{\textsc{XXI}, \textsc{XZI}, \textsc{ZIX}, \textsc{ZIZ}\}$ and $\{\textsc{XIX}, \textsc{XIZ}, \textsc{ZXI}, \textsc{ZZI}\}$ where we use the notation $\textsc{X} \equiv \sigma_x$, $\textsc{Z} \equiv \sigma_z$ and $\textsc{I} \equiv \sigma_0$. By Lemma~\ref{lem:corr-comp}, the sum of squares of the four correlation tensor elements in each set is bounded by unity giving the CHSH monogamy relation (\ref{eq:chsh-mono-corr}). 
 
 One method to derive trade-off relations between $k$ sets of $n$-party correlation inequalities is therefore to list the $k \cdot 2^n$ correlation tensor elements appearing in the inequality (\ref{eq:gen-bell-2}) and group them into sets of mutually anti-commuting elements. If a grouping into $k$ such sets exists, then a monogamy relation holds between the inequalities. In what follows, we investigate qubit networks represented by hypergraphs where such trade-off relations exist.

 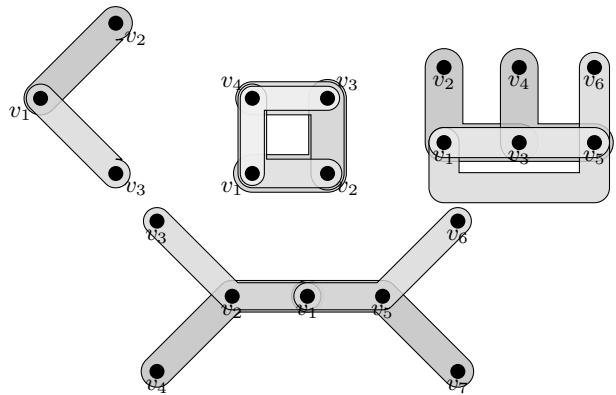
\begin{figure}
 	\centering
 			\begin{tikzpicture} 
 	
 	\node (v1) at (0, 0) {};
 	\node (v2) at (1, 1) {};
 	\node (v3) at (1, -1) {};
 	
 	\begin{scope}[fill opacity=0.8]
 	\filldraw[fill=gray!50]
 	($(v1) + ({-0.22 / sqrt(2)}, {0.22 / sqrt(2)})$) 
 	to[out=45,in=225] ($(v2) + ({-0.22 / sqrt(2)}, {0.22 / sqrt(2)})$)
 	to[out=45,in=180] ($(v2) + (0, 0.22)$)
 	to[out=0,in=180] ($(v2) + (0, 0.22)$)
 	to[out=0,in=90] ($(v2) + (0.22, 0)$)
 	to[out=270,in=0] ($(v2) + (0, -0.22)$)
 	to[out=180,in=0] ($(v2) + ({0.22 * (sqrt(2) - 1)}, -0.22)$)
 	to[out=225,in=45] ($(v1) + ({0.22 / sqrt(2)}, {-0.22 / sqrt(2)})$)
 	to[out=225,in=315] ($(v1) + ({-0.22 / sqrt(2)}, {-0.22 / sqrt(2)})$)
 	to[out=135,in=225] ($(v1) + ({-0.22 / sqrt(2)}, {0.22 / sqrt(2)})$)
 	;
 	\filldraw[fill=gray!30]
 	($(v1) + ({-0.19 / sqrt(2)}, {-0.19 / sqrt(2)})$) 
 	to[out=315,in=135] ($(v3) + ({-0.19 / sqrt(2)}, {-0.19 / sqrt(2)})$)
 	to[out=315,in=180] ($(v3) + (0, -0.19)$)
 	to[out=0,in=180] ($(v3) + (0, -0.19)$)
 	to[out=0,in=270] ($(v3) + (0.19, 0)$)
 	to[out=90,in=0] ($(v3) + (0, 0.19)$)
 	to[out=180,in=0] ($(v3) + ({0.19 * (sqrt(2) - 1)}, 0.19)$)
 	to[out=135,in=315] ($(v1) + ({0.19 / sqrt(2)}, {0.19 / sqrt(2)})$)
 	to[out=135,in=45] ($(v1) + ({-0.19 / sqrt(2)}, {0.19 / sqrt(2)})$)
 	to[out=225,in=135] ($(v1) + ({-0.19 / sqrt(2)}, {-0.19 / sqrt(2)})$)
 	;
 	\end{scope}
 	
 	\fill (v1) circle (0.1) node [below left] {$v_1$};
 	\fill (v2) circle (0.1) node [below right] {$v_2$};
 	\fill (v3) circle (0.1) node [below right] {$v_3$};
 	
 	
 	\end{tikzpicture}
 	\qquad
 	\begin{tikzpicture} 
 	\node (vA) at (0, 0) {};
 	\node (vB) at (1, 0) {};
 	\node (vC) at (1, 1) {};
 	\node (vD) at (0, 1) {};
 	
 	\begin{scope}[fill opacity=0.8]
 	
 	\filldraw[fill=gray!40]
 	($(vA) + ({0}, {-0.25})$) 
 	to[out=0,in=180] ($(vB) + ({0}, {-0.25})$)
 	to[out=0,in=270] ($(vB) + ({0.25}, {0})$)
 	to[out=90,in=270] ($(vC) + ({0.25}, {0})$)
 	to[out=90,in=0] ($(vC) + ({0}, {0.25})$)
 	to[out=180,in=90] ($(vC) + ({-0.25}, {0})$)
 	to[out=270,in=90] ($(vB) + ({-0.25}, {0.25})$)
 	to[out=180,in=0] ($(vA) + ({0}, {0.25})$)
 	to[out=180,in=90] ($(vA) + ({-0.25}, {0})$)
 	to[out=270,in=180] ($(vA) + ({0}, {-0.25})$)
 	;
 	
 	\filldraw[fill=gray!30]
 	($(vB) + ({0.22}, {0})$) 
 	to[out=90,in=270] ($(vC) + ({0.22}, {0})$)
 	to[out=90,in=0] ($(vC) + ({0}, {0.22})$)
 	to[out=180,in=0] ($(vD) + ({0}, {0.22})$)
 	to[out=180,in=90] ($(vD) + ({-0.22}, {0})$)
 	to[out=270,in=180] ($(vD) + ({0}, {-0.22})$)
 	to[out=0,in=180] ($(vC) + ({-0.22}, {-0.22})$)
 	to[out=270,in=90] ($(vB) + ({-0.22}, {0})$)
 	to[out=270,in=180] ($(vB) + ({0}, {-0.22})$)
 	to[out=0,in=270] ($(vB) + ({0.22}, {0})$)
 	;
 	
 	\filldraw[fill=gray!20]
 	($(vB) + ({0}, {-0.19})$) 
 	to[out=180,in=0] ($(vA) + ({0}, {-0.19})$)
 	to[out=180,in=270] ($(vA) + ({-0.19}, {0})$)
 	to[out=90,in=270] ($(vD) + ({-0.19}, {0})$)
 	to[out=90,in=180] ($(vD) + ({0}, {0.19})$)
 	to[out=0,in=90] ($(vD) + ({0.19}, {0})$)
 	to[out=270,in=90] ($(vA) + ({0.19}, {0.19})$)
 	to[out=0,in=180] ($(vB) + ({0}, {0.19})$)
 	to[out=0,in=90] ($(vB) + ({0.19}, {0})$)
 	to[out=270,in=0] ($(vB) + ({0}, {-0.19})$)
 	;
 	
 	\filldraw[fill=gray!10]
 	($(vA) + ({-0.16}, {0})$) 
 	to[out=90,in=270] ($(vD) + ({-0.16}, {0})$)
 	to[out=90,in=180] ($(vD) + ({0}, {0.16})$)
 	to[out=0,in=180] ($(vC) + ({0}, {0.16})$)
 	to[out=0,in=90] ($(vC) + ({0.16}, {0})$)
 	to[out=270,in=0] ($(vC) + ({0}, {-0.16})$)
 	to[out=180,in=0] ($(vD) + ({0.16}, {-0.16})$)
 	to[out=270,in=90] ($(vA) + ({0.16}, {0})$)
 	to[out=270,in=0] ($(vA) + ({0}, {-0.16})$)
 	to[out=180,in=270] ($(vA) + ({-0.16}, {0})$)
 	;
 	
 	\end{scope}
 	
 	\fill (vA) circle (0.1) node [below left] {$v_1$};
 	\fill (vB) circle (0.1) node [below right] {$v_{2}$};
 	\fill (vC) circle (0.1) node [above right] {$v_{3}$};
 	\fill (vD) circle (0.1) node [above left] {$v_{4}$};
 	
 	\end{tikzpicture}
 	\qquad
 	\begin{tikzpicture} 
 	
 	\node (v1) at (0, 0) {};
 	\node (v2) at (0, 1) {};
 	\node (v3) at (1, 0) {};
 	\node (v4) at (1, 1) {};
 	\node (v5) at (2, 0) {};
 	\node (v6) at (2, 1) {};
 	
 	\begin{scope}[fill opacity=0.8]
 	
 	\filldraw[fill=gray!50]
 	($(v1) + ({0}, {-0.25})$) 
 	to[out=0,in=180] ($(v3) + ({0}, {-0.25})$)
 	to[out=0,in=270] ($(v3) + ({0.25}, {0})$)
 	to[out=90,in=0] ($(v3) + ({0}, {0.25})$)
 	to[out=180,in=0] ($(v1) + ({0.25}, {0.25})$)
 	to[out=90,in=270] ($(v2) + ({0.25}, {0})$)
 	to[out=90,in=0] ($(v2) + ({0}, {0.25})$)
 	to[out=180,in=90] ($(v2) + ({-0.25}, {0})$)
 	to[out=270,in=90] ($(v1) + ({-0.25}, {0})$)
 	to[out=270,in=180] ($(v1) + ({0}, {-0.25})$)
 	;
 	
 	\filldraw[fill=gray!50]
 	($(v3) + ({0}, {-0.25})$) 
 	to[out=0,in=180] ($(v5) + ({0}, {-0.25})$)
 	to[out=0,in=270] ($(v5) + ({0.25}, {0})$)
 	to[out=90,in=0] ($(v5) + ({0}, {0.25})$)
 	to[out=180,in=0] ($(v3) + ({0.25}, {0.25})$)
 	to[out=90,in=270] ($(v4) + ({0.25}, {0})$)
 	to[out=90,in=0] ($(v4) + ({0}, {0.25})$)
 	to[out=180,in=90] ($(v4) + ({-0.25}, {0})$)
 	to[out=270,in=90] ($(v3) + ({-0.25}, {0})$)
 	to[out=270,in=180] ($(v3) + ({0}, {-0.25})$)
 	;
 	
 	\filldraw[fill=gray!30]
 	($(v6) + ({0.2}, {0})$) 
 	to[out=270,in=90] ($(v5) + ({0.2}, {-3*0.2})$)
 	to[out=270,in=0] ($(v5) + ({0}, {-4*0.2})$)
 	to[out=180,in=0] ($(v1) + ({0}, {-4*0.2})$)
 	to[out=180,in=270] ($(v1) + ({-0.2}, {-3*0.2})$)
 	to[out=90,in=270] ($(v1) + ({-0.2}, {0})$)
 	to[out=90,in=180] ($(v1) + ({0}, {0.2})$)
 	to[out=0,in=90] ($(v1) + ({0.2}, {0})$)
 	to[out=270,in=90] ($(v1) + ({0.2}, {-2*0.2})$)
 	to[out=0,in=180] ($(v5) + ({-0.2}, {-2*0.2})$)
 	to[out=90,in=270] ($(v6) + ({-0.2}, {0})$)
 	to[out=90,in=180] ($(v6) + ({0}, {0.2})$)
 	to[out=0,in=90] ($(v6) + ({0.2}, {0})$)
 	;
 	
 	\filldraw[fill=gray!15]
 	($(v1) + ({-0.2}, {0})$) 
 	to[out=90,in=180] ($(v1) + ({0}, {0.2})$)
 	to[out=0,in=180] ($(v5) + ({0}, {0.2})$)
 	to[out=0,in=90] ($(v5) + ({0.2}, {0})$)
 	to[out=270,in=0] ($(v5) + ({0}, {-0.2})$)
 	to[out=180,in=0] ($(v1) + ({0}, {-0.2})$)
 	to[out=180,in=270] ($(v1) + ({-0.2}, {0})$)
 	;
 	
 	\end{scope}
 	
 	\fill (v1) circle (0.1) node [below] {$v_1$};
 	\fill (v2) circle (0.1) node [below] {$v_{2}$};
 	\fill (v3) circle (0.1) node [below] {$v_{3}$};
 	\fill (v4) circle (0.1) node [below] {$v_{4}$};
 	\fill (v5) circle (0.1) node [below] {$v_{5}$};
 	\fill (v6) circle (0.1) node [below] {$v_{6}$};
 	
 	\end{tikzpicture}
 	\qquad 
 	 	\begin{tikzpicture} 
 	
 	\node (v1) at (0, 0) {};
 	\node (v2) at (-1, 0) {};
 	\node (v3) at (-2, 1) {};
 	\node (v4) at (-2, -1) {};
 	\node (v5) at (1, 0) {};
 	\node (v6) at (2, 1) {};
 	\node (v7) at (2, -1) {};
 	
 	\begin{scope}[fill opacity=0.8]
 	
 	\filldraw[fill=gray!50]
 	($(v4) + ({-0.21 / sqrt(2)}, {0.21 / sqrt(2)})$) 
 	to[out=45,in=225] ($(v2) + ({-0.21 / sqrt(2)}, {0.21 / sqrt(2)})$)
 	to[out=45,in=180] ($(v2) + (0, {0.21})$)
 	to[out=0,in=180] ($(v1) + (0, {0.21})$)
 	to[out=0,in=90] ($(v1) + ({0.21}, 0)$)
 	to[out=270,in=0] ($(v1) + (0, {-0.21})$)
 	to[out=180,in=0] ($(v2) + ({(sqrt(2)-1) * 0.21}, {-0.21})$)
 	to[out=225,in=45] ($(v4) + ({0.21 / sqrt(2)}, {-0.21 / sqrt(2)})$)
 	to[out=225,in=315] ($(v4) + ({-0.21 / sqrt(2)}, {-0.21 / sqrt(2)})$)
 	to[out=135,in=225] ($(v4) + ({-0.21 / sqrt(2)}, {0.21 / sqrt(2)})$)
 	;
 	
 	\filldraw[fill=gray!50]
 	($(v1) + (0, 0.21)$) 
 	to[out=0,in=180] ($(v5) + (0, {0.21})$)
 	to[out=0,in=135] ($(v5) + ({0.21 / sqrt(2)}, {0.21 / sqrt(2)})$)
 	to[out=315,in=135] ($(v7) + ({0.21 / sqrt(2)}, {0.21 / sqrt(2)})$)
 	to[out=315,in=45] ($(v7) + ({0.21 / sqrt(2)}, {-0.21 / sqrt(2)})$)
 	to[out=225,in=315] ($(v7) + ({-0.21 / sqrt(2)}, {-0.21 / sqrt(2)})$)
 	to[out=135,in=315] ($(v5) + ({-(sqrt(2)-1) * 0.21}, {-0.21})$)
 	to[out=180,in=0] ($(v1) + ({0}, {-0.21})$)
 	to[out=180,in=270] ($(v1) + ({-0.21}, {0})$)
 	to[out=90,in=180] ($(v1) + ({0}, {0.21})$)
 	;
 	
 	\filldraw[fill=gray!30]
 	($(v3) + ({-0.18 / sqrt(2)}, {-0.18 / sqrt(2)})$) 
 	to[out=315,in=135] ($(v2) + ({-0.18 / sqrt(2)}, {-0.18 / sqrt(2)})$)
 	to[out=315,in=180] ($(v2) + (0, {-0.18})$)
 	to[out=0,in=180] ($(v1) + (0, {-0.18})$)
 	to[out=0,in=270] ($(v1) + ({0.18}, 0)$)
 	to[out=90,in=0] ($(v1) + (0, {0.18})$)
 	to[out=180,in=0] ($(v2) + ({(sqrt(2)-1) * 0.18}, {0.18})$)
 	to[out=135,in=315] ($(v3) + ({0.18 / sqrt(2)}, {0.18 / sqrt(2)})$)
 	to[out=135,in=45] ($(v3) + ({-0.18 / sqrt(2)}, {0.18 / sqrt(2)})$)
 	to[out=225,in=135] ($(v3) + ({-0.18 / sqrt(2)}, {-0.18 / sqrt(2)})$)
 	;
 	
 	\filldraw[fill=gray!30]
 	($(v1) + (0, -0.18)$)
 	to[out=0,in=180] ($(v5) + (0, {-0.18})$)
 	to[out=0,in=225] ($(v5) + ({0.18 / sqrt(2)}, {-0.18 / sqrt(2)})$)
 	to[out=45,in=225] ($(v6) + ({0.18 / sqrt(2)}, {-0.18 / sqrt(2)})$)
 	to[out=45,in=315] ($(v6) + ({0.18 / sqrt(2)}, {0.18 / sqrt(2)})$)
 	to[out=135,in=45] ($(v6) + ({-0.18 / sqrt(2)}, {0.18 / sqrt(2)})$)
 	to[out=225,in=45] ($(v5) + ({-(sqrt(2)-1) * 0.18}, {0.18})$)
 	to[out=180,in=0] ($(v1) + ({0}, {0.18})$)
 	to[out=180,in=90] ($(v1) + ({-0.18}, {0})$)
 	to[out=270,in=180] ($(v1) + ({0}, {-0.18})$)
 	;	
 	
 	\end{scope}
 	
 	\fill (v1) circle (0.1) node [below] {$v_1$};
 	\fill (v2) circle (0.1) node [below] {$v_{2}$};
 	\fill (v3) circle (0.1) node [below] {$v_{3}$};
 	\fill (v4) circle (0.1) node [below] {$v_{4}$};
 	\fill (v5) circle (0.1) node [below] {$v_{5}$};
 	\fill (v6) circle (0.1) node [below] {$v_{6}$};
 	\fill (v7) circle (0.1) node [below] {$v_{7}$};
 	
 	\end{tikzpicture}
 	
 	\caption{Examples of qubit networks where tight trade-offs arise due to the uncertainty relation for anti-commuting observables (see also~\cite{KPRLK11} where similar trade-offs were derived). For any four qubit state and an arbitrary $3$-party full correlation inequality (binary \textsc{XOR} game) $\mathcal{I}$ , it holds that $\langle \mathcal{I} \rangle^2_{123}+ \langle \mathcal{I} \rangle^2_{234}+ \langle \mathcal{I} \rangle^2_{341}+ \langle \mathcal{I} \rangle^2_{412}\leq 4$. For any $6$ qubit state , it holds that	$\langle \mathcal{I} \rangle^2_{123}+\langle \mathcal{I} \rangle^2_{345}+ \langle \mathcal{I} \rangle^2_{561}+ \langle \mathcal{I} \rangle^2_{135}\leq 4$. For any $7$ qubit state, it holds that $\langle \mathcal{I} \rangle^2_{123}+ \langle \mathcal{I} \rangle^2_{124}+ \langle \mathcal{I} \rangle^2_{156}+ \langle \mathcal{I} \rangle^2_{157}\leq 4.$}
 	\label{fig:example-networks}
 \end{figure}

  Examples of hypergraphs with monogamy relations are given in Fig.~\ref{fig:example-networks}. A mathematical technique to derive relations of this form is to represent the observables $\mathcal{O} := \otimes_{l=1}^n O^{(l)}_{k_{l}}$ corresponding to the correlation tensor elements $T_{k_{1}, \dots, k_n}$ as vertices $v_{\mathcal{O}} \in V(\Gamma)$ of a graph $\Gamma$ with an edge between two vertices denoting anti-commutation of the associated operators i.e., $v_{\mathcal{O}} \sim v_{\mathcal{O}'} \Leftrightarrow \{\mathcal{O}, \mathcal{O'}\} = 0$. Equivalently, following the literature on error-correcting codes, one could also represent an operator $\mathcal{O}$ (with $\mathcal{O} = \exp{(i \phi)} \; \textbf{X}^\textbf{a} \textbf{Z}^\textbf{b}$, $\textbf{X}^\textbf{a} = X^{a_1} \otimes X^{a_2} \otimes \dots \otimes X^{a_n}$) as a $2n$-bit binary vector $(\textbf{a}| \textbf{b})$. The vertices then correspond to $2n$-bit binary strings and two vertices are connected by an edge if their symplectic inner product is $1$, where the symplectic inner product $\odot$ is defined by $(\textbf{a}| \textbf{b}) \odot (\textbf{a'} | \textbf{b'}) = \textbf{a} \cdot \textbf{b'} + \textbf{a'} \cdot \textbf{b}$, where $\cdot$ denotes the usual binary inner product. Remarkably, this exact structure has been studied in the mathematical literature as the \textit{unitary self-adjoint representation of a graph}. The problem of grouping these operators into sets of mutually anti-commuting ones then becomes equivalent to the problem of finding a clique cover of the graph, i.e., partitioning the vertices of the graph into cliques (recall that a clique in a graph is a subset of mutually adjacent vertices). 
 
 We first analyze the structure of the anti-commutation graphs that arise from hyperspherical monogamy relations of the type (\ref{eq:gen-monogamy-relation}). In particular, we observe properties of the anti-commutation graph in relation to its clique size (maximal number of mutually adjacent vertices) and vertex clique cover number (smallest number of cliques needed to cover all the vertices of the graph). We will restrict our attention to the scenario of a uniform hypergraph of degree $r(H) = \log{(\vert E(H) \vert)} + 1$ as in the CHSH monogamy relation (\ref{eq:CHSH-mono}). Note that each hyperedge in $H$ corresponds to a Bell expression so that the total number of Bell expressions in the trade-off relation is $|E(H)|$, and the degree $r(H)$ denotes the number of parties participating in a single expression with each expression involving the same number of parties due to the uniformity of the hypergraph. 
%
 \begin{prop}
 	The anti-commutation graph $\Gamma$ of a hypersphere monogamy relation arising from correlation complementarity for an $n$-party correlation Bell inequality arranged in a configuration given by a uniform hypergraph $H$ with degree  $r(H) = \log{(\vert E(H) \vert)} + 1$ is characterized by: (i) $|V(\Gamma)| = 2 |E(H)|^2$, (ii) clique size $\omega_c(\Gamma) = 2|E(H)|$, (iii) vertex clique cover number $cp(\Gamma) = |E(H)|$ and (iv) regularity with degree $2^{2(n-1)}-1$.
 \end{prop}
\begin{proof}
	The first property is immediate: (i) each hyperedge corresponds to $2^{r(H)}$ operators from Eq.(\ref{eq:gen-bell-2}), the hypergraph being uniform each edge contains the same number of vertices. The second property follows from the fact that to derive a monogamy relation we group the operators into sets of size $2^{r(H)}$ which is a clique in the graph. To see that the clique is of maximal size, it suffices to observe that within each Bell expression in (\ref{eq:gen-bell-2}) there are at most two mutually anti-commuting operators. The vertex clique cover number $cp(\Gamma) = |E(H)|$ follows from $|V(\Gamma)|$ and $\omega(\Gamma)$. Each operator from a given Bell expression anti-commutes with $2^{r(H)-1}-1$ operators from within the same expression due to the fact that the tensor products of Pauli operators either commute or anti-commute and each operator commutes with exactly half of the others. Similarly, the operator anti-commutes with exactly half of the operators from each of the other Bell expressions (recall that in order to have $\omega(\Gamma) = 2 |E(H)|$ each operator must anti-commute with at least two operators from every other Bell expression). 
\end{proof}

We now proceed to derive the smallest size of a qubit network in which one might expect to derive monogamy relations among $(n,2,2)$ correlation inequalities for given $n$ using the formalism of anti-commutation graphs. 
 
 \begin{prop}
 	\label{prop:min-qub-network}
 	The smallest size of a qubit network for which a hypersphere monogamy relation for the $n$-party correlation Bell expression can be derived by the method outlined above is $2^{n-1}$. 
 \end{prop}
 \begin{proof}
 	The unitary self-adjoint representations of anti-commutation graphs have been studied by Samoilenko~\cite{Sam91}. The following lemma was shown for the irreducible unitary representations in a given dimension . 
 	\begin{lemma}[\cite{Sam91}]
 		\label{lem:unitary-rep-graph}
 		An $n$-vertex graph $\Gamma$ has $2^{p(\Gamma)}$ $(0 \leq p(\Gamma) \leq n)$ unitarily inequivalent irreducible unitary self-adjoint representations of equal dimension $2^{m(\Gamma)}$ where $2m(\Gamma)+p(\Gamma)=n$. 
 	\end{lemma}
 The exact values of $m(\Gamma)$ and $p(\Gamma)$ for a given graph $\Gamma$ can be computed using a construction used in the proof of the Lemma~\ref{lem:unitary-rep-graph} in~\cite{Sam91}. 
 
 Now, in deriving monogamy relations using correlation complementarity, we encounter graphs with clique size $\omega(\Gamma) = 2 |E(H)| = 2^n$ for an $(n,2,2)$ inequality. The proof rests on the fact that the smallest dimension of the Hilbert space in which one can have a unitary self-adjoint representation of the complete graph $K_m$ on $m$ vertices (for even $m$) is $2^{m/2}$. Let us firstly show that such a representation of $K_m$ exists in dimension $2^{m/2}$. We pick an edge $(v_1, v_2) \in E(K_m)$, the associated binary observables $O_1$ and $O_2$ in the Hilbert space $\mathcal{H}$ anti-commute. We decompose $\mathcal{H} = \mathbb{C}^2 \otimes \mathcal{H}_1$ and set $O_1 = \sigma_x \otimes \mathbf{1}_{\mathcal{H}_1}$ and $O_2 = \sigma_z \otimes \mathbf{1}_{\mathcal{H}_1}$. Let us now consider every other vertex $v_k$ ($3 \leq k \leq m$) adjacent to $v_1$ and $v_2$. Since the associated $O_k$ is required to anti-commute with both $O_1$ and $O_2$, we assign $O_k = \sigma_y \otimes O^{(1)}_k$. The representation problem for the graph $K_m$ is now reduced to that for the graph $K_{m-2}$ obtained from $K_m$ by deleting the vertices $v_1$ and $v_2$. We proceed as above until we arrive at the single edge $K_2$ at which point we end with the operators $O_{n-1} = \left( \otimes_{i=1}^{m/2-1} \sigma_y \right) \otimes \sigma_x$ and $O_{n} = \left( \otimes_{i=1}^{m/2-1} \sigma_y \right) \otimes \sigma_z$. Note also by Lemma~\ref{lem:unitary-rep-graph}, the above construction gives the unique unitary self-adjoint representation of the complete graph $K_m$ on $m$ vertices in dimension $2^{m/2}$ up to unitary equivalence. Moreover, by the proof of Lemma~\ref{lem:unitary-rep-graph} every unitary self-adjoint representation of $K_m$ in any dimension up to $2^{m/2}$ is obtainable up to unitary equivalence by the above construction if it exists. This gives $p(K_m) = 0$ so that no unitary self-adjoint representation of $K_m$ exists in a lower-dimensional Hilbert space. Therefore, the smallest dimension in which such a representation exists for the clique of size $2^n$ is $2^{2^{n-1}}$ proving the claim.   
\end{proof}
 
 An example of a qubit network saturating the minimal size imposed by Proposition~\ref{prop:min-qub-network} for $n=3$ is the square network from Fig.~\ref{fig:example-networks} with the associated trade-off relation Eq.(\ref{eq:square-network}).
 
 \textit{Iterative constructions of qubit networks obeying tight trade-off relations.-}

We will now provide a way to derive infinite families of trade-off relations using a generalization of the family of Bell operators introduced in~\cite{KPRLK11}. We start with providing a couple of examples and then formalize the construction and prove the trade-offs. The method of~\cite{KPRLK11} allowed to derive the $7$-qubit tree network from Fig.~\ref{fig:example-networks} for which
 \begin{eqnarray}
 \label{TH_QUBOUND7}
 \langle \mathcal{I} \rangle^2_{123}+ \langle \mathcal{I} \rangle^2_{124}+ \langle \mathcal{I} \rangle^2_{156}+ \langle \mathcal{I} \rangle^2_{157}\leq 4.
 \end{eqnarray}
 The $32$ operators occurring in~\eqref{TH_QUBOUND7} are divided into $4$ groups of $8$ operators on a Hilbert space $\mathcal{H}$ of dimension $2^7$ (the size of the qubit network is $7$). The next member of this family consists of $8$ groups of $16$ anti-commuting operators, and gives the trade-off:
 \begin{eqnarray}
 &&\langle \mathcal{I} \rangle^2_{1,2,3,4}+ \langle \mathcal{I} \rangle^2_{1,2,3,5}+ \langle \mathcal{I} \rangle^2_{1,2,6,7}+ \langle \mathcal{I} \rangle^2_{1,2,6,8} + \langle \mathcal{I} \rangle^2_{1,9,10,11}\nonumber \\
 &&+\langle \mathcal{I} \rangle^2_{1,9,11,12}+\langle \mathcal{I} \rangle^2_{1,9,13,14}+ \langle \mathcal{I} \rangle^2_{1,9,13,15}
 \leq 8.
 \end{eqnarray}

One may check that using Lemma~\ref{lem:corr-comp} the following monogamy relation holds:
 \begin{align}
 \label{eq:square-network}
 \langle \mathcal{I} \rangle^2_{123}+ \langle \mathcal{I} \rangle^2_{234}+ \langle \mathcal{I} \rangle^2_{341}+ \langle \mathcal{I} \rangle^2_{412}\leq 4,
 \end{align}
which follows from dividing the operators in $4$ groups of $8$ anti-commuting elements.
The construction given below uses existing sets of anti-commuting operators to joint them in a similar way as in~\eqref{TH_QUBOUND7} to get larger groupings of operators. For example, one may take two trivial groups, $\{\sigma_1\}$ and $\{\sigma_2\}$, and the grouping leading to \eqref{eq:square-network} to derive the following trade-off:
 \begin{eqnarray}
 &&\langle \mathcal{I} \rangle^2_{1234}+ \langle \mathcal{I} \rangle ^2_{1345}+ \langle \mathcal{I} \rangle^2_{1452}+ \langle \mathcal{I} \rangle^2_{1523} +
 \langle \mathcal{I} \rangle^2_{1678} \nonumber \\
 &&\qquad +\langle \mathcal{I} \rangle^2_{1789}+ \langle \mathcal{I} \rangle^2_{1896}+ \langle \mathcal{I} \rangle^2_{1967}
 \leq 8.
 \end{eqnarray}
The construction thus gives a qubit network of size $9$ with $8$ groups of $16$ operators each. An example of such a grouping is \small{XIIIIXXXI, XIIIIYYYI, XIIIIXYIX, XIIIIYXIY, XIIIIXIYY, XIIIIYIXX, XIIIIIXYX, XIIIIIYXY} together with \small{YXXXIIIII, YYYYIIIII, YXYIXIIII, YYXIXIIII, YXIYYIIII, YYIXXIIII, YIXYXIIII, YIYXYIIII}.
 
 Now, let us formalize the construction. For a set $S$ let $\{\Pi^{S}_{i}: S \rightarrow S\}_{i=1}^{\Ab{S}}$ be an arbitrary set of permutations over the set $S$ with a property that for any pair $s_1,s_2 \in S$ there exists exactly one value $i$ such that $\Pi^{S}_{i}(s_1) = s_2$. Let $[n] \equiv \{1, \cdots, n\}$.
 
 Let $\Omega^{(X)}$ be a set of self-adjoint operators on a Hilbert space $\mathcal{H}^{(X)}$, and $\{\omega^{(X)}_{i}\}_{i=1}^{N^{(X)}}$ be a disjoint partition (grouping) of $\Omega^{(X)}$, i.e. $\omega^{(X)}_{i} \cap \omega^{(X)}_{j} = \emptyset$ and $\Omega^{(X)} = \bigcup_{i=1}^{N^{(X)}} \omega^{(X)}_{i}$, where $\Ab{\omega^{(X)}_{i}} = n^{(X)}$ and operators within each set $\omega^{(X)}_{i}$ anti-commute.
 
 Let us consider a Hilbert space $\mathcal{H} = \mathcal{H}^{(0)} \otimes \bigotimes_{i=1}^{n^{(0)}} \mathcal{H}^{(i)}$ with defined sets of self-adjoint operators $\Omega^{(i)}$ grouped into groups of size $n^{(i)}$, $\{\omega^{(i)}_{j}\}_{j=1}^{N^{(j)}}$. Using these operators one can construct the following set of operators on $\mathcal{H}$:
 \begin{equation}
 \Omega \equiv \bigcup_{k \in [n^{(0)}]} \bigcup_{i \in [n^{(0)} \cdot N^{(0)}]} \bigcup_{j \in [n^{(k)} \cdot N^{(k)}]} \{T_{i,j,k} \},
 \end{equation}
 where $T_{i,j,k} \equiv \sigma^{(0)}_{i} \otimes \sigma^{(k)}_{j}$. We have $\Ab{\Omega} = n^{(0)} \cdot N^{(0)} \cdot N \cdot \left( \sum_{l \in [n^{(0)}]} n^{(l)} \right)$.
 
 Let us assume that for all $i \geq 1$ we have $N^{(i)} = N$ for some $N \geq 1$. We show that the operators from the set $\Omega$ can be grouped into $n^{(0)} \cdot N^{(0)} \cdot N$ disjoint sets of anti-commuting operators, each of the size $\sum_{l \in [n^{(0)}]} n^{(l)}$. Indeed, for $i \in [n^{(0)}]$, $j \in [N^{(0)}]$ and $k \in [N]$ let us define the set:
 \begin{equation}
 \omega_{i,j,k} \equiv \bigcup_{l \in [n^{(0)}]} \bigcup_{\sigma^{(l)} \in \omega^{(l)}_{k}} \left\{ \left( \Pi^{\omega^{(0)}_{j}}_{i} \left[ \sigma^{(0)}_{j,l} \right] \right) \otimes \sigma^{(l)} \right\} \subset \Omega,
 \end{equation}
 where $\sigma^{(m)}_{j,l}$ denotes $l$-th operator in the set $\omega^{(m)}_{j}$ (for some fixed ordering). We have $\Ab{\omega_{i,j,k}} = \sum_{l \in [n^{(0)}]} n^{(l)}$. It is easy to see that all operator within this set anti-commute since they either differ at the space $\mathcal{H}^{(0)}$ and commute at other spaces, or they have the same operator at the space $\mathcal{H}^{(0)}$ and anti-commute at $\mathcal{H}^{(l)}$, for some $l$.
 
 The sets $\omega_{i,j,k}$ are disjoint. Indeed, for different values of $j$ and $k$, the operators $\sigma^{(0)}_{j,l}$ and $\sigma^{(l)}$, respectively, are take from different groups. For different values of $i$ the operators on $\mathcal{H}^{(l)}$ are related with different operators on $\mathcal{H}^{(0)}$. Thus, comparing the sizes of sets we get
 \begin{equation}
 \Omega = \bigcup_{i \in [n^{(0)}]} \bigcup_{j \in [N^{(0)}]} \bigcup_{k \in [N]} \omega_{i,j,k}.
 \end{equation}

\textit{Bell monogamy relations beyond correlation complementarity.-}
We have seen how trade-off relations for correlation Bell inequalities in certain qubit networks can be derived by means of the complementarity principle of anti-commuting observables. In this section, we consider trade-off relations arising from beyond this complementarity principle and analyze an explicit ladder network where we show that a hyperspherical monogamy relation holds.

\begin{figure}[h !]
	\centering
		\begin{tikzpicture}
			
			\node (v1) at (0, 0) {};
			\node (v2) at (1, 1) {};
			\node (v3) at (1, -1) {};
			\node (v4) at (3, 1) {};
			\node (v5) at (3, -1) {};
			
			\begin{scope}[fill opacity=0.8]
				\filldraw[fill=gray!50]
					($(v1) + ({-0.3 / sqrt(2)}, {0.3 / sqrt(2)})$) 
					to[out=45,in=225] ($(v1) + ({-0.3 / sqrt(2)}, {0.3 / sqrt(2)})$)
					to[out=45,in=225] ($(v2) + ({-0.3 / sqrt(2)}, {0.3 / sqrt(2)})$)
					to[out=45,in=135] ($(v2) + ({0.3 / sqrt(2)}, {0.3 / sqrt(2)})$)
					to[out=315,in=135] ($(v5) + ({0.3 / sqrt(2)}, {0.3 / sqrt(2)})$)
					to[out=315,in=45] ($(v5) + ({0.3 / sqrt(2)}, {-0.3 / sqrt(2)})$)
					to[out=225,in=315] ($(v5) + ({-0.3 / sqrt(2)}, {-0.3 / sqrt(2)})$)
					to[out=135,in=315] ($(v2) + (0, {-0.3 * sqrt(2)})$)
					to[out=225,in=45] ($(v1) + ({0.3 / sqrt(2)}, {-0.3 / sqrt(2)})$)
					to[out=225,in=315] ($(v1) + ({-0.3 / sqrt(2)}, {-0.3 / sqrt(2)})$)
					to[out=135,in=225] ($(v1) + ({-0.3 / sqrt(2)}, {0.3 / sqrt(2)})$)
					;
				\filldraw[fill=gray!50]
					($(v1) + ({-0.3 / sqrt(2)}, {-0.3 / sqrt(2)})$) 
					to[out=315,in=135] ($(v1) + ({-0.3 / sqrt(2)}, {-0.3 / sqrt(2)})$)
					to[out=315,in=135] ($(v3) + ({-0.3 / sqrt(2)}, {-0.3 / sqrt(2)})$)
					to[out=315,in=225] ($(v3) + ({0.3 / sqrt(2)}, {-0.3 / sqrt(2)})$)
					to[out=45,in=225] ($(v4) + ({0.3 / sqrt(2)}, {-0.3 / sqrt(2)})$)
					to[out=45,in=315] ($(v4) + ({0.3 / sqrt(2)}, {0.3 / sqrt(2)})$)
					to[out=135,in=45] ($(v4) + ({-0.3 / sqrt(2)}, {0.3 / sqrt(2)})$)
					to[out=225,in=45] ($(v3) + (0, {0.3 * sqrt(2)})$)
					to[out=135,in=315] ($(v1) + ({0.3 / sqrt(2)}, {0.3 / sqrt(2)})$)
					to[out=135,in=45] ($(v1) + ({-0.3 / sqrt(2)}, {0.3 / sqrt(2)})$)
					to[out=225,in=135] ($(v1) + ({-0.3 / sqrt(2)}, {-0.3 / sqrt(2)})$)
					;
				\filldraw[fill=gray!30]
					($(v1) + ({-0.22 / sqrt(2)}, {0.22 / sqrt(2)})$) 
					to[out=45,in=225] ($(v2) + ({-0.22 / sqrt(2)}, {0.22 / sqrt(2)})$)
					to[out=45,in=180] ($(v2) + (0, 0.22)$)
					to[out=0,in=180] ($(v4) + (0, 0.22)$)
					to[out=0,in=90] ($(v4) + (0.22, 0)$)
					to[out=270,in=0] ($(v4) + (0, -0.22)$)
					to[out=180,in=0] ($(v2) + ({0.22 * (sqrt(2) - 1)}, -0.22)$)
					to[out=225,in=45] ($(v1) + ({0.22 / sqrt(2)}, {-0.22 / sqrt(2)})$)
					to[out=225,in=315] ($(v1) + ({-0.22 / sqrt(2)}, {-0.22 / sqrt(2)})$)
					to[out=135,in=225] ($(v1) + ({-0.22 / sqrt(2)}, {0.22 / sqrt(2)})$)
					;
				\filldraw[fill=gray!30]
					($(v1) + ({-0.22 / sqrt(2)}, {-0.22 / sqrt(2)})$) 
					to[out=315,in=135] ($(v3) + ({-0.22 / sqrt(2)}, {-0.22 / sqrt(2)})$)
					to[out=315,in=180] ($(v3) + (0, -0.22)$)
					to[out=0,in=180] ($(v5) + (0, -0.22)$)
					to[out=0,in=270] ($(v5) + (0.22, 0)$)
					to[out=90,in=0] ($(v5) + (0, 0.22)$)
					to[out=180,in=0] ($(v3) + ({0.22 * (sqrt(2) - 1)}, 0.22)$)
					to[out=135,in=315] ($(v1) + ({0.22 / sqrt(2)}, {0.22 / sqrt(2)})$)
					to[out=135,in=45] ($(v1) + ({-0.22 / sqrt(2)}, {0.22 / sqrt(2)})$)
					to[out=225,in=135] ($(v1) + ({-0.22 / sqrt(2)}, {-0.22 / sqrt(2)})$)
					;
			\end{scope}
			
			\fill (v1) circle (0.1) node [below right] {$v_1$};
			\fill (v2) circle (0.1) node [below right] {$v_2$};
			\fill (v3) circle (0.1) node [below right] {$v_3$};
			\fill (v4) circle (0.1) node [below right] {$v_4$};
			\fill (v5) circle (0.1) node [below right] {$v_5$};
			
			
			\path[dotted]
			(v4) edge node[above] {} +(2,0)
			(v4) edge node[above] {} +(2,-2)
			(v5) edge node[above] {} +(2,0)
			(v5) edge node[above] {} +(2,2)
			;

		\end{tikzpicture}
	\caption{The ladder network consisting of $2n-1$ qubits (shown for $n=3$ in the figure). The ladder depicts a single player Alice playing an arbitrary binary \textsc{XOR} game with two Bobs, two Charlies, etc. The sum of squares of the quantum values that the players can achieve in the $2^{n-1}$ games (depicted by the hyper-edges in the graph) obeys a tight trade-off relation that goes beyond the uncertainty relation for anti-commuting observables.}
	\label{fig:ladder}
\end{figure}
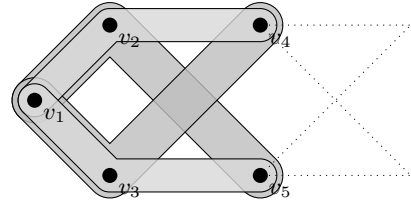

\begin{prop}
	\label{prop:ladder-network}
	Let $n \in \mathbb{Z}$ with $n \geq 2$. Consider the ladder network shown in Fig.~\ref{fig:ladder} of $2n -1$ qubits. For any $n$-partite full-correlation Bell inequality $\mathcal{I}$ there holds 
	\begin{eqnarray}
		\sum_{e \in \mathcal{E}} \langle \mathcal{I} \rangle^2_e \leq 2^{n-1}. 
	\end{eqnarray}
	This relation is tight.
\end{prop}

\begin{proof}
Let $\sigma_{\vec{l}} \equiv \bigotimes \sigma_{l_i}$, $l_i \in \{0,1,2,3\}$. Let $S_{1} \equiv \{(1), (3)\}$ and $S_{2n+1} \equiv S_{2n-1} \otimes \{(0,1),(0,3),(1,0),(3,0)\} = \{\vec{l}\}$, so that e.g.
\begin{eqnarray}
	S_{3} = \{(1,0,1), (1,0,3), (1,1,0), (1,3,0), \nonumber \\
	(3,0,1), (3,0,3), (3,1,0), (3,3,0)\}.
\end{eqnarray}
Let us also define the operators $\mu_{1} := \sigma_0 \otimes \sigma_1$, $\mu_2 := \sigma_0 \otimes \sigma_3$, $\mu_3 := \sigma_1 \otimes \sigma_0$ and $\mu_4 := \sigma_3 \otimes \sigma_0$. Thus, we have:
\begin{equation}
 \label{eq:inductive-S}
 \sum_{\vec{l} \in S_{2k+1}} \sigma_{\vec{l}} = \sum_{\vec{l} \in S_{2k-1}} \sum_{r=1}^{4} \sigma_{\vec{l}} \otimes \mu_k.
\end{equation}

Let $\ket{\psi^{(n)}}$ denote a state in $2 n$ dimensional Hilbert space. We start the proof by noting the following simple fact:
\begin{equation}
 \begin{aligned}
	 & \bigforall_{ \ket{\psi^{(1)}_1}, \ket{\psi^{(1)}_2} \in \mathbb{C}^{2} } \bra{\psi^{(1)}_1,\psi^{(1)}_2} \sum_{k \in S_{1}} \sigma_k \otimes \sigma_k \ket{\psi^{(1)}_2, \psi^{(1)}_1} = \\
	 & \Ab{\bra{\psi^{(1)}_1} \sigma_1 \ket{\psi^{(1)}_2}}^2 + \Ab{\bra{\psi^{(1)}_1} \sigma_3 \ket{\psi^{(1)}_2}}^2 \leq 1.
 \end{aligned}
\end{equation}
In particular for $\ket{\psi^{(1)}_1}, \ket{\psi^{(1)}_2} \in \mathbb{R}^{2}$ the equality holds. This can be seen by a direct computation.

Now we proceed by induction to prove the monogamy relation for larger values of $n$. We assume that the following holds for some $n$:
\begin{widetext}
\begin{equation}
 \label{eq:inductive-eqn}
 \max_{ \ket{\psi^{(2n-1)}_1}, \ket{\psi^{(2n-1)}_2} \in \mathbb{C}^{2n-1} } \bra{\psi^{(2n-1)}_1, \psi^{(2n-1)}_2} \left( \sum_{\vec{l} \in S_{2n-1}} \sigma_{\vec{l}} \otimes \sigma_{\vec{l}} \right) \ket{\psi^{(2n-1)}_2, \psi^{(2n-1)}_1} = 2^{k-1}.
\end{equation}
\end{widetext}

Let us now derive the following Lemma. 
\begin{lemma}
	\label{lem:maxVectors}
	For any set $S$ indexing an arbitrary set of Hermitian operators $\{ M_k \}_{k \in S}$ acting on a Hilbert space $\mathcal{H}$ we have:
	\begin{eqnarray}
	\label{eq:maxVectors}
	&& \max_{\ket{\phi_1}, \ket{\phi_2}, \ket{\phi_3}, \ket{\phi_4} \in \mathcal{H}} \sum_{k \in S} \bra{\phi_1} M_k \ket{\phi_2} \bra{\phi_3} M_k \ket{\phi_4} = \nonumber \\
	&& \max_{\ket{\phi_1}, \ket{\phi_2} \in \mathcal{H}} \bra{\phi_1, \phi_2} \sum_{k \in S} M_k \otimes M_k \ket{\phi_2, \phi_1}.
	\end{eqnarray}
	\begin{proof}
		For given $\ket{\phi_1}, \ket{\phi_2}, \ket{\phi_3}, \ket{\phi_4} \in \mathcal{H}$ let $\vec{u}, \vec{v} \in \mathbb{C}^{|S|}$ be defined by $u_k \equiv \bra{\phi_2} M_k \ket{\phi_1}$ and $v_k \equiv \bra{\phi_3} M_k \ket{\phi_4}$. The LHS of the equation~\eqref{eq:maxVectors} can be rewritten as $\max_{\vec{u},\vec{v}} \vec{u}^{\dagger} \cdot \vec{v}$. From Cauchy-Schwarz inequality we see that the maximum is attained when $\vec{u} = \vec{v}$, and thus we have $\ket{\phi_1} = \ket{\phi_4}$ and $\ket{\phi_2} = \ket{\phi_3}$. After relabeling and noting that $\bra{\phi_1} M_k \ket{\phi_2} \bra{\phi_2} M_k \ket{\phi_1} = \bra{\phi_1, \phi_2} M_k \otimes M_k \ket{\phi_2, \phi_1}$, this yields the equality~\eqref{eq:maxVectors}.
	\end{proof}
\end{lemma}

Note that an arbitrary $(2n+1)$-qubit state $\ket{\psi^{(2n+1)}_{i}} \in \mathbb{C}^{2(2n+1)}$ can be written as 
\begin{equation}
	\label{eq:inductive-state}
	\ket{\psi^{(2n+1)}_{i}} = \sum_{j, k = 0,1} \beta_{i,j,k} \ket{\psi^{(2n-1)}_{i,j,k}} \ket{j,k},  
\end{equation}
with $\beta_{i,j,k} \in \mathbb{R}$, $\sum_{j, k = 0, 1} \beta_{i,j,k}^2 = 1$ for all $i$, $\ket{\psi^{(2n-1)}_{i,j,k}} \in \mathbb{C}^{2(2n-1)}$ and $\{\ket{j,k}\}$ being the computational basis of $\mathbb{C}^2 \otimes \mathbb{C}^2$.

Using the Lemma~\ref{lem:maxVectors} and the inductive equations~\eqref{eq:inductive-S} and~\eqref{eq:inductive-eqn} we have:
\begin{widetext}
	\begin{equation}
		\begin{aligned}
			& \max_{ \ket{\psi^{(2n+1)}_1}, \ket{\psi^{(2n+1)}_2}, \ket{\psi^{(2n+1)}_3}, \ket{\psi^{(2n+1)}_4} \in \mathbb{C}^{2(2n+1)} } \bra{\psi^{(2n+1)}_1, \psi^{(2n+1)}_2} \left( \sum_{\vec{l} \in S_{2n+1}} \sigma_{\vec{l}} \otimes \sigma_{\vec{l}} \right) \ket{\psi^{(2n+1)}_3, \psi^{(2n+1)}_4} \\
			& = \max_{ \ket{\psi^{(2n+1)}_1}, \ket{\psi^{(2n+1)}_2} \in \mathbb{C}^{2(2n+1)} } \bra{\psi^{(2n+1)}_1, \psi^{(2n+1)}_2} \left( \sum_{\vec{l} \in S_{2n+1}} \sigma_{\vec{l}} \otimes \sigma_{\vec{l}} \right) \ket{\psi^{(2n+1)}_2, \psi^{(2n+1)}_1} \\
			& = 2^{n-1} \max_{\{\beta_{i,j,k}\}} \sum_{\substack{{j_1, j_2, j_3, j_4} \\ {k_1, k_2, k_3, k_4}} = 0,1} \beta_{1,j_1,k_1} \beta_{2,j_2,k_2}  \beta_{2,j_3,k_3} \beta_{1,j_4,k_4} \bra{j_1, k_1, j_2, k_2} \sum_{r=1}^4 \mu_r \otimes \mu_r \ket{j_3, k_3, j_4, k_4} \\
			& = 2^{n-1} \max_{\{\beta_{i,j,k}\}} \left[ 2 - 2 p_1^2 - p_2^2 - p_3^2 \right] = 2^n,
		\end{aligned}
	\end{equation}
\end{widetext}
where
	\begin{eqnarray}
	p_1 = \beta_{1,0,0}\beta_{2,1,1} -\beta_{1,0,1}\beta_{2,1,0} -\beta_{1,1,0}\beta_{2,0,1} +\beta_{1,1,1}\beta_{2,0,0}, \nonumber \\
	p_2 = \beta_{1,0,0}\beta_{2,1,0} +\beta_{1,0,1}\beta_{2,1,1} -\beta_{1,1,0}\beta_{2,0,0} -\beta_{1,1,1}\beta_{2,0,1}, \nonumber \\ 
	p_3 = \beta_{1,0,0}\beta_{2,0,1} -\beta_{1,0,1}\beta_{2,0,0} +\beta_{1,1,0}\beta_{2,1,1} -\beta_{1,1,1}\beta_{2,1,0}. \nonumber \\
	\end{eqnarray}
Assuming $\beta_{i,j,k} = \beta_{j,k}$ we get $p_1 = 2(\beta_{0,0} \beta_{1,1} - \beta_{0,1} \beta_{1,0})$, $p_2 = p_3 = 0$.

It is easy to see that the above maximum can be attained with $\ket{\psi^{(2n+1)}_1} = \ket{\psi^{(2n+1)}_2} \in \mathbb{R}^{2n+1}$. e.g. if we take $\ket{\psi^{(1)}_{i,j,k}} \in \mathbb{R}^2$ and $\beta_{0,0} \beta_{1,1} = \beta_{0,1} \beta_{1,0}$ for all induction steps~\eqref{eq:inductive-state}. This shows the claimed monogamy relation for the ladder network in Fig.~\ref{fig:ladder}. 
\end{proof}

As an example of the above trade-off relation, we have
\begin{eqnarray}
&&\langle \text{Mermin} \rangle^2_{A,B^{(1)},C^{(1)}} + \langle \text{Mermin} \rangle^2_{A,B^{(1)},C^{(2)}} + \langle \text{Mermin} \rangle^2_{A,B^{(2)},C^{(1)}} \nonumber \\
&& \qquad \qquad \qquad \qquad + \langle \text{Mermin} \rangle^2_{A,B^{(2)},C^{(2)}} \leq 16, \nonumber \\
\end{eqnarray}
for the usual Mermin inequality given as
\begin{eqnarray}
&& \langle \text{Mermin} \rangle_{AB^{(1)}C^{(1)}} = \langle -A_1 B^{(1)}_1 C^{(1)}_1  +  A_1 B^{(1)}_2 C^{(1)}_2 +  \nonumber \\
&&\qquad \qquad A_2 B^{(1)}_1 C^{(1)}_2 +  A_2 B^{(1)}_2 C^{(1)}_1 \rangle \leq \beta_c,
\end{eqnarray}
where $\beta_c = 2$ is the classical bound and $\beta_q = 4$ is the quantum value. 

A brute force search over the operators appearing in the trade-off relation in Prop.~\ref{prop:ladder-network} (for small values of $n$) reveals that there do not exist anti-commuting sets of sufficient size to imply this tight trade-off relation. 
Interestingly, the above hyperspherical relation has a hyperplane analogue in general no-signaling theories, providing an exact generalization of the CHSH monogamy found by Toner and Verstraete in~\cite{TV06}. Namely, we have 
\begin{prop} 
	\label{prop:ladder-ns}
	Let $n \in \mathbb{Z}$ with $n \geq 2$. Consider the ladder network shown in Fig.~\ref{fig:ladder} of $2n-1$ qubits. For any $n$-partite full-correlation Bell inequality $\mathcal{I}$ there holds in all no-signaling theories the following trade-off relation
	\begin{eqnarray}
	\sum_{e \in E} \langle \mathcal{I} \rangle \leq 2^{n-1}. 
	\end{eqnarray}
\end{prop}
The above proposition is a corollary of the more general result in the $(n,m,d)$ setting shown in Prop.~\ref{prop:genBell-ns-mono}.

\textit{Tightness of the trade-off relations.-}
Having derived general trade-off relations for correlation inequalities in qubit networks, we proceed to investigate the tightness of these relations. When the bound on the sum of several distinct Bell expressions is saturated by a quantum strategy, we say that the relation is tight. When a quantum strategy exists to achieve every possible tuple of the Bell values saturating the trade-off relation, we say that the relation is spherically tight. 
\begin{prop} 
A monogamy relation involving $k$-partite Bell parameters $\mathcal{I}_e$ for odd $k$ is spherically tight if and only if it is of the form
\begin{equation}
\sum_{e = 1}^m \langle \mathcal{I} \rangle_e^2 \le 2^{k-1} \quad \textrm{ for } \quad 2 \le m \le 2^{k-1}.
\label{STBM}
\end{equation}
\end{prop}
\begin{proof}
	Consider a spherically tight trade-off relation of the form:
	\begin{equation}
	\sum_{e = 1}^m \langle \mathcal{I} \rangle_e^2 \leq \beta_c,
	\end{equation}
	for some constant $\beta_c$. Since the possible algebraic values of $\langle \mathcal{I} \rangle^2_e$ that are allowed range from $0$ to $2^{k-1}$, $\beta_c$ is at most $2^{k-1}$. This is simply due to the fact that a larger value of $\beta_c$ would imply that setting the value of all but one of the Bell expressions in the trade-off to $0$, the remaining expression could not achieve the bound set by $\beta_c$ as required for a spherically tight relation. That the maximum of any single $\langle \mathcal{I} \rangle^2_e$ is at most $2^{k-1}$ follows from the fact that the Bell expression for $\langle \mathcal{I} \rangle^2_e$ can be rewritten as the sum of squares of $2^k$ correlation tensor elements, which can be grouped into $2^{k-1}$ pairs of anti-commuting elements. Moreover, all values from 0 to $2^{k-1}$ are realized by some quantum states (notably the maximum value of $2^{k-1}$ arises in the well-known algebraic violation of the Mermin inequalities~\cite{MERMIN1990}), thus $\beta_c \geq 2^{k-1}$. Therefore $\beta_c =2^{k-1}$ and all spherically tight Bell monogamies are of the form (\ref{STBM}).
	
	We prove that that the monogamy relation~\ref{STBM} is spherically tight by considering the following state
	\begin{equation}
	\ket{\psi} = \frac{1}{\sqrt{2}} \sum_{e=1}^m \alpha_e \ket{\underbrace{0 \dots 0}_{v} 1 \dots 1} + \frac{1}{\sqrt{2}} \ket{1 \dots 1},
	\end{equation}
	where the zeros in the summed kets are at the positions of the $k$ qubits involved in the $e$-th Bell parameter,
	and we choose positive reals $\alpha_e$ such that $\alpha_1^2 + \dots + \alpha_m^2 = 1$.
	It is convenient to define $\ket{e} \equiv \ket{0 \dots 0 1 \dots 1}$ having zeros for the parties of the $e$-th Bell parameter.
	For this state, the bound on every single Bell inequality is given by the sum of squared correlation tensor elements $T_{j_1 \dots j_k 0 \dots 0} = \bra{\psi} \sigma_{j_1} \otimes \dots \otimes \sigma_{j_k} \ket{\psi}$ in the $xy$ plane of the correlation tensor.
	Since Pauli operators $\sigma_x$ and $\sigma_y$ flip the eigenstates of the $\sigma_z$ operator,
	the state $\sigma_{j_1} \otimes \dots \otimes \sigma_{j_k} \ket{e}$ is equal up to a global phase to at most one other state $\ket{e'}$.
	The equality can only happen if the number of $1$'s after the flip is the same as before the flip, i.e., exactly half of the $0$'s in $\ket{e}$ are flipped.
	Therefore, if $k$ is odd, this equality cannot happen and we conclude that every single Bell expression 
	obeys $\langle \mathcal{I} \rangle_e^2 = 2^{k-1} \alpha_e^2$.
	The spherical tightness then arises by the normalization of the state $| \psi \rangle$.
\end{proof}

\textit{Quantum monogamies for binary XOR games played by different numbers of parties.-}

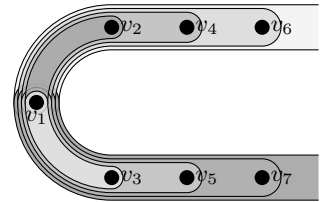
\begin{figure}[b !]
	\centering
	\begin{tikzpicture}
	
	\node (v1) at (0, 0) {};
	\node (v2) at (1, 1) {};
	\node (v3) at (1, -1) {};
	\node (v4) at (2, 1) {};
	\node (v5) at (2, -1) {};
	\node (v6) at (3, 1) {};
	\node (v7) at (3, -1) {};
	
	\begin{scope}[fill opacity=0.8]
	\filldraw[fill=gray!10] 
	($(v6) + (0.75, -0.30)$)
	to[out=180,in=0] ($(v2) + (0, -0.30)$)
	to[out=180,in=90] ($(v1) + (0.30, 0)$)
	to[out=270,in=0] ($(v1) + (0, -0.30)$)
	to[out=180,in=270] ($(v1) + (-0.30, 0)$)
	to[out=90,in=180] ($(v2) + (0, 0.30)$)
	to[out=0,in=180] ($(v6) + (0.75, 0.30)$)
	;
	\filldraw[fill=gray!80] 
	($(v7) + (0.75, 0.30)$)
	to[out=180,in=0] ($(v3) + (0, 0.30)$)
	to[out=180,in=270] ($(v1) + (0.30, 0)$)
	to[out=90,in=0] ($(v1) + (0, 0.30)$)
	to[out=180,in=90] ($(v1) + (-0.30, 0)$)
	to[out=270,in=180] ($(v3) + (0, -0.30)$)
	to[out=0,in=180] ($(v7) + (0.75, -0.30)$)
	;
	\filldraw[fill=gray!30] 
	($(v1) + (-0.25, 0)$)
	to[out=90,in=180] ($(v2) + (0, 0.25)$)
	to[out=0,in=180] ($(v6) + (0, 0.25)$)
	to[out=0,in=90] ($(v6) + (0.25, 0)$)
	to[out=270,in=0] ($(v6) + (0, -0.25)$)
	to[out=180,in=0] ($(v2) + (0, -0.25)$)
	to[out=180,in=90] ($(v1) + (0.25, 0)$)
	to[out=270,in=0] ($(v1) + (0, -0.25)$)
	to[out=180,in=270] ($(v1) + (-0.25, 0)$)
	;
	\filldraw[fill=gray!60] 
	($(v1) + (-0.25, 0)$)
	to[out=270,in=180] ($(v3) + (0, -0.25)$)
	to[out=0,in=180] ($(v7) + (0, -0.25)$)
	to[out=0,in=270] ($(v7) + (0.25, 0)$)
	to[out=90,in=0] ($(v7) + (0, 0.25)$)
	to[out=180,in=0] ($(v3) + (0, 0.25)$)
	to[out=180,in=270] ($(v1) + (0.25, 0)$)
	to[out=90,in=0] ($(v1) + (0, 0.25)$)
	to[out=180,in=90] ($(v1) + (-0.25, 0)$)
	;
	\filldraw[fill=gray!50] 
	($(v1) + (-0.2, 0)$)
	to[out=90,in=180] ($(v2) + (0, 0.2)$)
	to[out=0,in=180] ($(v4) + (0, 0.2)$)
	to[out=0,in=90] ($(v4) + (0.2, 0)$)
	to[out=270,in=0] ($(v4) + (0, -0.2)$)
	to[out=180,in=0] ($(v2) + (0, -0.2)$)
	to[out=180,in=90] ($(v1) + (0.2, 0)$)
	to[out=270,in=0] ($(v1) + (0, -0.2)$)
	to[out=180,in=270] ($(v1) + (-0.2, 0)$)
	;
	\filldraw[fill=gray!40] 
	($(v1) + (-0.2, 0)$)
	to[out=270,in=180] ($(v3) + (0, -0.2)$)
	to[out=0,in=180] ($(v5) + (0, -0.2)$)
	to[out=0,in=270] ($(v5) + (0.2, 0)$)
	to[out=90,in=0] ($(v5) + (0, 0.2)$)
	to[out=180,in=0] ($(v3) + (0, 0.2)$)
	to[out=180,in=270] ($(v1) + (0.2, 0)$)
	to[out=90,in=0] ($(v1) + (0, 0.2)$)
	to[out=180,in=90] ($(v1) + (-0.2, 0)$)
	;
	\filldraw[fill=gray!70] 
	($(v1) + (-0.15, 0)$) 
	to[out=90,in=180] ($(v2) + (0, 0.15)$)
	to[out=0,in=90] ($(v2) + (0.15, 0)$)
	to[out=270,in=0] ($(v2) + (0, -0.15)$)
	to[out=180,in=90] ($(v1) + (0.15, 0)$)
	to[out=270,in=0] ($(v1) + (0, -0.15)$)
	to[out=180,in=270] ($(v1) + (-0.15, 0)$)
	;
	\filldraw[fill=gray!20] 
	($(v1) + (-0.15, 0)$) 
	to[out=270,in=180] ($(v3) + (0, -0.15)$)
	to[out=0,in=270] ($(v3) + (0.15, 0)$)
	to[out=90,in=0] ($(v3) + (0, 0.15)$)
	to[out=180,in=270] ($(v1) + (0.15, 0)$)
	to[out=90,in=0] ($(v1) + (0, 0.15)$)
	to[out=180,in=90] ($(v1) + (-0.15, 0)$)
	;
	\end{scope}
	
	\fill (v1) circle (0.1) node [below] {$v_1$};
	\fill (v2) circle (0.1) node [right] {$v_2$};
	\fill (v3) circle (0.1) node [right] {$v_3$};
	\fill (v4) circle (0.1) node [right] {$v_4$};
	\fill (v5) circle (0.1) node [right] {$v_5$};
	\fill (v6) circle (0.1) node [right] {$v_6$};
	\fill (v7) circle (0.1) node [right] {$v_7$};
	
	\end{tikzpicture}
	\caption{A network of $2n-1$ qubits in which a tight trade-off relation holds between the quantum values of binary \textsc{XOR} games played by subsets of different sizes as shown in Prop.~\ref{prop:diff-no-players}.}
	\label{fig:propSum}
\end{figure}

\begin{prop}
	\label{prop:diff-no-players}
	Consider $2n-1$ parties arranged in the configuration represented by the hypergraph in Fig.~\ref{fig:propSum}. Let $\mathcal{I}^{(k)}_{l_1, \dots, l_k}$ denote a $k$-party binary XOR game played by the players $l_1, \dots, l_k \in [2n-1]$. The following trade-off relation holds for the value of any such game within quantum theory
	\begin{eqnarray}
	& \left \langle \mathcal{I}^{(2)} \right \rangle^2_{1,2} + \left \langle \mathcal{I}^{(2)} \right \rangle^2_{1,3} + \left \langle \mathcal{I}^{(3)} \right \rangle^2_{1,2,4} + \left \langle \mathcal{I}^{(3)} \right \rangle^2_{1,3,5} + \nonumber \\
	& \left \langle \mathcal{I}^{(4)} \right \rangle^2_{1,2,4,6} + \left \langle \mathcal{I}^{(4)} \right \rangle^2_{1,3,5,7} + \dots + \\
	& \left \langle \mathcal{I}^{(n)} \right \rangle^2_{1,2,4,\dots,2n-2} + \left \langle \mathcal{I}^{(n)} \right \rangle^2_{1,3,5,\dots,2n-1} \leq 2^{n-1}. \nonumber
	\end{eqnarray}
\end{prop}

\begin{proof}
	Explicitly, we are required to prove the following bound on the correlation tensor elements for any $2n-1$ qubit state $|\psi \rangle \in (\mathbb{C}^2)^{\otimes 2n-1}$ in the network configuration represented by the hypergraph in Fig.~\ref{fig:propSum}
	\begin{widetext}
		\begin{eqnarray}
		&&\sum_{k_1, k_2 =1,2} T_{k_1, k_2, 0, \dots, 0}^2 + \sum_{k_1, k_3 = 1,2} T_{k_1, 0, k_3, 0, \dots, 0}^2 +  \sum_{k_1, k_2, k_4 = 1,2} T_{k_1, k_2, 0, k_4, 0,\dots, 0}^2 + \sum_{k_1, k_3, k_5=1,2} T_{k_1, 0, k_3, 0, k_5, 0, \dots, 0}^2 \nonumber \\ &&+  \dots + \sum_{k_1, k_2, k_4, \dots, k_{2n-2} = 1,2} T_{k_1, k_2, 0, k_4, 0, k_6, \dots, 0, k_{2n-2}}^2 + \sum_{k_1, k_3, k_5, \dots, k_{2n-1} = 1,2} T_{k_1, 0, k_3, 0, k_5, 0, \dots, 0, k_{2n-1}}^2  \leq 2^{n-1}.
		\end{eqnarray}
	\end{widetext}
	Here, we encounter a list of $2 \times (2^2 + 2^3 + \dots + 2^{n}) = 8(2^{n-1} - 1)$ operators which we would like to group into $2^{n-1}$ sets of mutually anti-commuting operators. We do this by splitting the $8(2^{n-1} -1)$ operators into $2^{n-l+1}$ sets of $2l$ anti-commuting operators each for $3 \leq l \leq n$ and an additional $2$ sets of $2n$ anti-commuting operators (exploiting the identities  $\sum_{l=3}^{n} 2^{n-l+1} 2l + 4n \equiv 8(2^{n-1} -1)$ and $\sum_{l=3}^{n} 2^{n-l+1}  = 2^{n-1} - 2$). 
	
	We explicitly identify one set of anti-commuting operators in the figure~\ref{fig:ACset-diff-parties}, to be read as 
	\begin{center}
		\begin{tabular}{ | l | c | }
			\hline
			\textsc{XXIXI\dots XIXIXI} & \textsc{ZIXIX\dots IXIXIX} \\ \hline
			\textsc{XXIXI\dots XIXIZI} & \textsc{ZIXIX\dots IXIXIZ}  \\ \hline
			\textsc{XXIXI\dots XIZIII} & \textsc{ZIXIX\dots IXIZII}  \\ \hline
			\; \textsc{XXIXI\dots ZIIIII} &  \textsc{ZIXIX\dots IZIIII} \\ \hline
			\; \; \; \; \; \dots &  \dots \\ \hline
			\; \textsc{XXIZI\dots IIIIII} & \textsc{ZIXIZ \dots IIIIII} \\ \hline
			\;  \textsc{XZIII\dots IIIIII} & \textsc{ZIZII \dots IIIIII} \\ 
			\hline
		\end{tabular}
	\end{center}

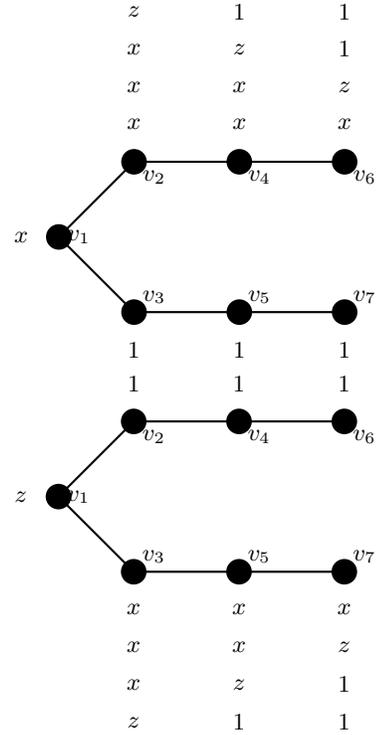
\begin{figure}
	\begin{center}
		\begin{tikzpicture}
		
		\node (v1) at (0, 0) {};
		\node (v2) at (1, 1) {};
		\node (v3) at (1, -1) {};
		\node (v4) at (2.4, 1) {};
		\node (v5) at (2.4, -1) {};
		\node (v6) at (3.8, 1) {};
		\node (v7) at (3.8, -1) {};
		
		\fill (v1) circle (0.17) node [right] {$v_1$};
		\fill (v2) circle (0.17) node [below right] {$v_2$};
		\fill (v3) circle (0.17) node [above right] {$v_3$};
		\fill (v4) circle (0.17) node [below right] {$v_4$};
		\fill (v5) circle (0.17) node [above right] {$v_5$};
		\fill (v6) circle (0.17) node [below right] {$v_6$};
		\fill (v7) circle (0.17) node [above right] {$v_7$};
		
		\node at ($ (v1) + (-0.5, 0) $) {$x$};
		
		\node at ($ (v2) + (0, 0.5) $) {$x$};
		\node at ($ (v2) + (0, 1) $) {$x$};
		\node at ($ (v2) + (0, 1.5) $) {$x$};
		\node at ($ (v2) + (0, 2) $) {$z$};
		
		\node at ($ (v4) + (0, 0.5) $) {$x$};
		\node at ($ (v4) + (0, 1) $) {$x$};
		\node at ($ (v4) + (0, 1.5) $) {$z$};
		\node at ($ (v4) + (0, 2) $) {$1$};
		
		\node at ($ (v6) + (0, 0.5) $) {$x$};
		\node at ($ (v6) + (0, 1) $) {$z$};
		\node at ($ (v6) + (0, 1.5) $) {$1$};
		\node at ($ (v6) + (0, 2) $) {$1$};
		
		\node at ($ (v3) + (0, -0.5) $) {$1$};
		\node at ($ (v5) + (0, -0.5) $) {$1$};
		\node at ($ (v7) + (0, -0.5) $) {$1$};
		
		\path[draw,thick]
		(v1) edge node {} (v2)
		(v1) edge node {} (v3)
		(v2) edge node {} (v4)
		(v4) edge node {} (v6)
		(v3) edge node {} (v5)
		(v5) edge node {} (v7)
		;
		
		\end{tikzpicture}
		\qquad
		\begin{tikzpicture}
		
		\node (v1) at (0, 0) {};
		\node (v2) at (1, 1) {};
		\node (v3) at (1, -1) {};
		\node (v4) at (2.4, 1) {};
		\node (v5) at (2.4, -1) {};
		\node (v6) at (3.8, 1) {};
		\node (v7) at (3.8, -1) {};
		
		\fill (v1) circle (0.17) node [right] {$v_1$};
		\fill (v2) circle (0.17) node [below right] {$v_2$};
		\fill (v3) circle (0.17) node [above right] {$v_3$};
		\fill (v4) circle (0.17) node [below right] {$v_4$};
		\fill (v5) circle (0.17) node [above right] {$v_5$};
		\fill (v6) circle (0.17) node [below right] {$v_6$};
		\fill (v7) circle (0.17) node [above right] {$v_7$};
		
		\node at ($ (v1) + (-0.5, 0) $) {$z$};
		
		\node at ($ (v3) + (0, -0.5) $) {$x$};
		\node at ($ (v3) + (0, -1) $) {$x$};
		\node at ($ (v3) + (0, -1.5) $) {$x$};
		\node at ($ (v3) + (0, -2) $) {$z$};
		
		\node at ($ (v5) + (0, -0.5) $) {$x$};
		\node at ($ (v5) + (0, -1) $) {$x$};
		\node at ($ (v5) + (0, -1.5) $) {$z$};
		\node at ($ (v5) + (0, -2) $) {$1$};
		
		\node at ($ (v7) + (0, -0.5) $) {$x$};
		\node at ($ (v7) + (0, -1) $) {$z$};
		\node at ($ (v7) + (0, -1.5) $) {$1$};
		\node at ($ (v7) + (0, -2) $) {$1$};
		
		\node at ($ (v2) + (0, 0.5) $) {$1$};
		\node at ($ (v4) + (0, 0.5) $) {$1$};
		\node at ($ (v6) + (0, 0.5) $) {$1$};
		
		\path[draw,thick]
		(v1) edge node {} (v2)
		(v1) edge node {} (v3)
		(v2) edge node {} (v4)
		(v4) edge node {} (v6)
		(v3) edge node {} (v5)
		(v5) edge node {} (v7)
		;
		
		\end{tikzpicture}
	\end{center}
	\caption{An explicit set of anti-commuting operators used in the proof of Prop.~\ref{prop:diff-no-players}. Following the numbering of the seven qubit example in the figure, the operators are to be read as \small{XXIXIXI, XXIXIZI, XXIZIII, XZIIIII, ZIXIXIX, ZIXIXIZ, ZIXIZII, ZIZIIII}.}
		\label{fig:ACset-diff-parties}
\end{figure}

	Each of the other sets is obtained from the first by interchanging $\textsc{X}_i  \leftrightarrow \textsc{Z}_i$ in all nodes of each leaf but the last node (where the leaves are the two sets of parties $\{2,4,6, \dots,\}$ and $\{3,5,7, \dots\}$), and deleting redundant operators (i.e., operators that occur more than once in the construction). Note that interchanging $\textsc{X}_i \leftrightarrow \textsc{Z}_i$ in all operators in a set preserves their commutation relations and that performing this interchange in the $n-1$ nodes barring the last node of each leaf gives $2^{n-1}$ sets exhausting the entire list of operators. By Lemma~\ref{lem:corr-comp} then, the sums of squares of the expectation values of the operators in each of the $2^{n-1}$ sets is bounded by unity, so that the claimed bound follows.

\end{proof}
Note that while the trade-off relations in the previous sections were spherically tight, the above trade-off is not spherically tight. Thus, while the bound of $2^{n-1}$ can be reached, for instance by a $n$-qubit GHZ state achieving the maximum value of $2^{(n-1)/2}$ for $\mathcal{I}^{(n)}_{1,2,4,\dots,2n-2}$ it is not the case that every tuple of values achieving the bound is realizable by a quantum strategy.

A special case of the above proposition is when $n = 3$. 
We then have a monogamy relation between bipartite and tripartite Bell parameters.
	\begin{eqnarray}
	\label{eq:CHSH-Merm}
	&&\langle \text{CHSH} \rangle_{AB^{(1)}}^2 +  \langle \text{CHSH} \rangle_{AB^{(2)}}^2 + \langle \text{Mermin} \rangle_{AB^{(1)} C^{(1)}}^2   \nonumber \\
	&& \qquad \qquad + \langle \text{Mermin} \rangle_{AB^{(2)} C^{(2)}}^2 \leq 16.
	\end{eqnarray}
To see Eq.(\ref{eq:CHSH-Merm}), observe that one can group the operators into four groups of six mutually anti-commuting ones as
\begin{eqnarray}
	\{\text{XXIXI, XXIZI, ZIXIX, ZIXIZ, XZIII, ZIZII}\}, \nonumber \\
	\{\text{XZIXI, XZIZI, ZIZIX, ZIZIZ, XXIII, ZIXII}\}, \nonumber \\
	\{\text{ZXIXI, ZXIZI, XIXIX, XIXIZ, ZZIII,XIZII}\}, \\
	\{\text{ZZIXI, ZZIZI, XIZIX, XIZIZ, ZXIII, XIXIX}\}. \nonumber
\end{eqnarray}
Each of the CHSH and Mermin expressions has a local bound of $2$ so that one can find a local box that achieves the bound of $16$ above. The square of a single Mermin expression has a quantum mechanical value of up to $16$ so that the bound can be saturated therein. However, note that unlike the trade-off relations in the previous sections, the bound above is not ``spherically tight", i.e., the corresponding linearized relation in spherical coordinates cannot be saturated for every value of the spherical angles. 

Interestingly, the no-signaling trade-off between the four Bell expressions is given as 
\begin{eqnarray}
&&\langle \text{CHSH} \rangle_{AB^{(1)}} + \langle \text{CHSH} \rangle_{AB^{(2)}} + \langle \text{Mermin} \rangle_{AB^{(1)} C^{(1)}} \nonumber \\
&&\qquad \qquad + \langle \text{Mermin} \rangle_{AB^{(2)} C^{(2)}} \leq 10.
\end{eqnarray}
This bound follows by direct computation by a linear program. Secret sharing~\cite{HBB98} is a task where a dealer (Alice) sends a secret S
to $n$ (possibly, dishonest) players so that the cooperation of a minimum of $k \leq n$ players is required to decode the secret. Protocols that accomplish this are called $(k,n)$-threshold schemes. 
Quantum Secret Sharing (QSS) schemes have been proposed to securely accomplish this task, by exploiting multipartite entanglement
to secure and split the classical secret among the players. 
In a QSS scheme, Alice's goal is to establish
a secret key with a joint degree of freedom of the players.
The players can only retrieve Alice's key and decode the
classical secret when they collaborate and communicate to each other their local measurements to form the joint variable. The trade-offs between Bell inequality violations with different numbers of parties established in Prop.~\ref{prop:diff-no-players} lends itself naturally as a security check for this task, a study which we pursue elsewhere.

\textit{Applications of the trade-off relations.-}


\textit{Flat non-local regions in the set of quantum correlations.-}
The above monogamy relations can evidently be used to derive Tsirelson-type bounds on the quantum values of the corresponding correlation inequalities. A linear version of the hyperspherical relations (obtained by parametrizing the relations in hyperspherical coordinates) also gives us novel inequalities whose maximal violation unlike the usual correlation inequalities is not achieved by a GHZ state. A natural question therefore arises whether such Bell expressions also serve to self-test the resulting optimal states. In this section we show that, in fact these new inequalities serve to provide first examples of Bell inequalities with the property that multiple quantum boxes attain their optimal violation, showing the presence of such flat regions in the set of quantum correlations. Since there are multiple distinct optimal quantum boxes (the probability tables in the different boxes do not match), these inequalities cannot serve as self-tests. In what follows, we restrict attention to the $(3,2,2)$ scenario consisting of three parties measuring two binary observables each, and denote the set of quantum correlations in this scenario as $\mathbf{Q}(3,2,2)$.

\begin{prop}
There exist non-trivial Bell inequalities in the $(3,2,2)$ scenario which are maximally violated by multiple distinct quantum boxes, i.e., there exist flat non-local regions in the boundary of the quantum set $\mathbf{Q}(3,2,2)$.
\end{prop}
\begin{proof}
We consider the Bell operator in the $(3,2,2)$ scenario of the form:
\begin{eqnarray}
\label{eq:non-unique-Bell}
\cos{\theta} \; \langle \text{CHSH} \rangle_{AB}  + \sin{\theta} \; \langle \text{CHSH} \rangle_{AC} \leq 2 \left(\cos{\theta} + \sin{\theta} \right) \leq 2 \sqrt{2}, \nonumber \\
\end{eqnarray}
where $\text{CHSH}_{AB} := A_1 B_1 + A_1 B_2 - A_2 B_1 + A_2 B_2$ and $\text{CHSH}_{AC} := A_1 C_1 + A_1 C_2 - A_2 C_1 + A_2 C_2$. The classical value of the inequality is given by $2(\cos{\theta} + \sin{\theta})$. The quantum value ($2 \sqrt{2}$) of the inequality in (\ref{eq:non-unique-Bell}) is obtained as a linearization of the spherical monogamy relation found by Toner and Verstraete~\cite{TV06}, i.e., for any quantum state and measurements it holds that 
\begin{eqnarray}
\label{eq:Toner-Verstraete-mono}
\langle \text{CHSH} \rangle_{AB}^2 + \langle \text{CHSH} \rangle_{AC}^2 \leq 8. 
\end{eqnarray}
The maximum classical value of the Bell expression is evidently $2(\cos{\theta} + \sin{\theta})$ following from the classical bound on the individual CHSH expressions. The maximum quantum value of the expression in (\ref{eq:non-unique-Bell}) is $2 \sqrt{2}$ for any value of $\theta$, this follows from (\ref{eq:Toner-Verstraete-mono}). The maximum value of the expression in all no-signaling theories is $4 \cos{\theta}$ for $\theta \in [0, \pi/4]$ and $4 \sin{\theta}$ for $\theta \in [\pi/4, \pi/2]$. This follows from the corresponding CHSH trade-off relation $\langle \text{CHSH} \rangle_{AB} + \langle \text{CHSH} \rangle_{AC} \leq 4$ in general no-signaling theories. Let us now exhibit two distinct boxes (together with the corresponding quantum state and measurements) that achieve the maximal quantum value of (\ref{eq:non-unique-Bell}).
\begin{enumerate}
\item For $\theta \in [0, \pi/4]$, measure on the state 
$| \psi_1 \rangle = \frac{\sqrt{1 - \sqrt{2} \sin{\theta}}}{2} \left(|010 \rangle + |011 \rangle \right) + \frac{\sqrt{1+ \sqrt{2} \sin{\theta}}}{2} \left( | 100 \rangle + |101 \rangle \right),$
the observables given by 
$A_1 = \sigma_x, \; \; \; \; A_2 = \sigma_{z}$
$B_1 = \cos{\phi_1} \; \sigma_x + \sin{\phi_1} \; \sigma_z, \; \; \; \; B_2 = \cos{\phi_2} \;  \sigma_x + \sin{\phi_2} \;  \sigma_z,$
$C_1 = \sigma_x, \; \; C_2 = - \sigma_x$,
with $\phi_1 = - \phi_2 = \phi = \arcsin{\frac{\sec{\theta}}{\sqrt{2}}}$. Direct calculation shows that the resulting box has the correlators
$\langle A_1 B_1 \rangle = \langle A_1 B_2 \rangle = \frac{\cos{2 \theta} \sec{\theta}}{\sqrt{2}}, \; \; \; \; \langle A_2 B_1 \rangle = - \langle A_2 B_2 \rangle = \frac{-\sec{\theta}}{\sqrt{2}},$
$\langle A_1 C_1 \rangle = \langle A_1 C_2 \rangle = 0, \; \;  \; \; \langle A_2 C_1 \rangle = - \langle A_2 C_2 \rangle = - \sqrt{2} \sin{\theta}.$
This achieves the values $\langle \text{CHSH} \rangle_{AB} = 2 \sqrt{2} \cos{\theta}$ and $\langle \text{CHSH} \rangle_{AC} = 2 \sqrt{2} \sin{\theta}$ and therefore the value $2 \sqrt{2}$ of (\ref{eq:non-unique-Bell}). Also, for $\theta \in [\pi/4, \pi/2]$, to get the same values we measure on the state 
$| \psi_2 \rangle = \frac{\sqrt{1 - \sqrt{2} \cos{\theta}}}{2} \left(|001 \rangle + |011 \rangle \right) + \frac{\sqrt{1+ \sqrt{2} \cos{\theta}}}{2} \left( | 100 \rangle + |110 \rangle \right),$
the observables given by
$A_1 = \sigma_{x}, A_2 = \sigma_{z},$
$B_1 = \sigma_{x}, B_2 = - \sigma_{x},$
$C_1 = \cos{\phi_1} \; \sigma_x + \sin{\phi_1} \; \sigma_z, C_2 = \cos{\phi_2} \; \sigma_x + \sin{\phi_2} \; \sigma_z,$
with $\phi_1 = - \phi_2 = \phi = \arcsin{\frac{\sec{\theta}}{\sqrt{2}}}$. 

\item A different quantum box which achieves the same values is given by the state
$| \chi \rangle = \frac{1}{\sqrt{2}} \left[ \cos{\theta} \; |001 \rangle + \sin{\theta} \; | 010 \rangle \right] + \frac{1}{\sqrt{2}} \; |111 \rangle$
and the measurements
$A_1 = \sigma_{x},  A_2 = \sigma_{y}$,
$B_1 = \frac{1}{\sqrt{2}} \left(\sigma_{x} + \sigma_{y} \right), B_2 = \frac{1}{\sqrt{2}} \left(\sigma_{x} - \sigma_{y} \right)$,
$C_1 = \frac{1}{\sqrt{2}} \left(\sigma_{x} + \sigma_{y} \right), C_2 = \frac{1}{\sqrt{2}} \left(\sigma_{x} - \sigma_{y} \right).$
Once again a calculation reveals that these measurements on $| \chi \rangle$ achieve the values 
$\langle A_1 B_1 \rangle = \langle A_1 B_2 \rangle = - \langle A_2 B_1 \rangle = \langle A_2 B_2 \rangle = \frac{\cos{\theta}}{\sqrt{2}},$
$\langle A_1 C_1 \rangle = \langle A_1 C_2 \rangle = - \langle A_2 C_1 \rangle = \langle A_2 C_2 \rangle = \frac{\sin{\theta}}{\sqrt{2}}.$
This gives $\langle \text{CHSH} \rangle_{AB} = 2 \sqrt{2} \cos{\theta}$ and $\langle \text{CHSH} \rangle_{AC} = 2 \sqrt{2} \sin{\theta}$ and therefore the maximum quantum value $2 \sqrt{2}$ of (\ref{eq:non-unique-Bell}). 
\end{enumerate}
Evidently any mixture of the boxes from (1) and (2) above also achieve the maximum quantum value $2 \sqrt{2}$ of the Bell expression (\ref{eq:non-unique-Bell}).  
\end{proof}

Note that the boxes (1) and (2) in the proof do not exhibit genuine tripartite nonlocality, as such an interesting open question is whether there exist flat regions in the quantum set where all the boxes in the region exhibit genuine three-party nonlocality.  

\textit{Bounds on the guessing probability.-}
The derivation of monogamy relations enables us to derive bounds on the guessing probability, an important quantity in device-independent cryptographic tasks such as randomness expansion and amplification~\cite{ADPTA14, AMP12, PAM+10}. A central aim in these tasks is to quantify the randomness generated by the boxes from the amount of Bell inequality violation seen by the honest parties alone, independently of the possible underlying quantum realizations compatible with this violation.
Accordingly, we model the initial state of the $n$ honest parties and Eve as $\rho_{\mathcal{A}_1,\dots,\mathcal{A}_n, \mathcal{E}}$ upon which they act with sets of measurement operators $\{M_{a_i|x_i}\}$ and $\{M_{e|z}\}$. After Alice's measurement $x^*_1$, the correlations between her classical output $a_1$ and the quantum information held by Eve are described by the classical-quantum state $\sum_{a_1} p_{\mathcal{A}_1}(a_1|x^*_1) | a_1 \rangle \langle a_1| \otimes \rho^{a_1,x^*_1}_{\mathcal{E}}$, with $\rho^{a_1,x^*_1}_{\mathcal{E}}$ being the reduced state of Eve given $x^*_1,a_1$. The guessing probability quantifies the randomness of Alice's output given this quantum side information of Eve, i.e., the probability that Eve correctly guesses Alice's output using an optimal strategy described by the POVM $\{M_{e|z}\}$. The guessing probability is thus given as
\begin{eqnarray}
&&P_{\text{guess}}(A_1|X_1 = x^*_1,\mathcal{E}) = \nonumber \\
&&\qquad\max_{Q, \{M_{e|z}\}} \sum_{a_1} P(A_1 = a_1|X_1 = x^*_1, Q)  \times \nonumber \\
&&\quad \qquad P(E=a_1| X_1 = x^*_1, Z=z,A_1 = a_1,Q ).
\end{eqnarray}
Here, $X_1,A_1$ are the input-output random variables of Alice's system, $Z,E$ are the input-output random variables of Eve's system, $\mathcal{E}$ denotes the side information held by Eve and $Q$ denotes a quantum box $\{P(a_1, \dots, a_n,e|x_1, \dots, x_n, z)\}$ whose marginal on the honest parties is compatible with the observed Bell violation. Note that for the case of binary outcomes, the guessing probability above is directly related to the correlation function between Alice's observable $X_1$ and Eve's optimal observable $Z$ as $\langle X_1 Z \rangle = 2 P_{\text{guess}}(A_1|X_1=x^*_1,\mathcal{E}) - 1$. The previously derived trade-off relations will enable us to bound the correlation function and hence to obtain a bound on Eve's guessing probability given an observed value of the multi-party Bell inequality violation. 

Before moving on to the general case, let us illustrate this by applying a  trade-off relation to derive a bound on the guessing probability of Alice's outcome upon the observed violation of a hybrid Mermin-CHSH inequality. Consider a Bell scenario with $4$ parties, $A^{(1)}, A^{(2)}, A^{(3)}$ being the honest parties and $A^{(4)}$ denoting the adversary Eve. A slight modification of the trade-off relation from Eq.(\ref{eq:CHSH-Merm}) gives that
\begin{eqnarray}
\langle \text{Mermin} \rangle^2_{1,2,3} + \langle \text{CHSH} \rangle^2_{1,2} + 2 \langle \text{CHSH} \rangle^2_{1,4} \leq 16.
\end{eqnarray}
For completeness, observe that the four sets of mutually anti-commuting operators is given as
\begin{eqnarray}
	\{\text{XXXI, XXZI, XZII, ZIIX, ZIIZ}\}, \nonumber \\
	\{\text{XZXI, XZZI, XXII, ZIIX, ZIIZ}\}, \nonumber \\
	\{\text{ZXXI, ZXZI, ZZII, XIIX, XIIZ}\}, \\
	\{\text{ZZXI, ZZZI, ZXII, XIIX, XIIZ}\}. \nonumber
\end{eqnarray}
To derive a bound on the guessing probability, we set $A^{(4)}_1 = A^{(4)}_2$ (where $A^{(4)}_1$ denotes Eve's optimal measurement) and obtain 
\begin{eqnarray}
\label{eq:pguess-merm-chsh}
&&P_{\text{guess}}(A_1|X_1 = 1, \mathcal{E}) \leq \nonumber \\
&&  \frac{1}{2} \left[ 1 + \sqrt{2 - (1/8) \langle \text{Mermin} \rangle^2_{1, 2, 3} - (1/8) \langle \text{CHSH} \rangle^2_{1,2} } \right].  
\end{eqnarray}
Upon observing the maximal violation of the Mermin inequality ($\langle \text{Mermin} \rangle_{1,2,3} = 4$), the guessing probability reduces to a random guess, while for the classical value ($ \langle \text{Mermin} \rangle_{1,2,3} = \langle \text{CHSH} \rangle_{1,2} =2$), the value of Alice's outcome can be perfectly guessed. 
The honest parties can thus check for the violation of an inequality of the form 
\begin{eqnarray}
\cos{(\theta)} \langle \text{Mermin} \rangle_{1,2,3} + \sin{(\theta)} \langle \text{CHSH} \rangle_{1,2}  \leq  2 \left(\cos{(\theta)} + \sin{(\theta)} \right), \nonumber \\
\end{eqnarray}
to use the bound (\ref{eq:pguess-merm-chsh}) in a device-independent application. 

The above considerations can be extended to derive bounds on the guessing probability in the case of $n$ honest parties. Following the proof of~\ref{prop:diff-no-players}
the resulting bound is then seen to be
\begin{widetext}
\begin{eqnarray}
\label{eq:pguess-merm-chsh}
P_{\text{guess}}(A_1|X_1 = 1, \mathcal{E}) \leq \frac{1}{2} \left[ 1 + \sqrt{2^{n-2} - (1/2^n) \langle \text{Mermin}^{(n)} \rangle^2_{1, 2, \dots, n} - \sum_{j = 2}^{n-1} (1/2^{j+1}) \langle \text{Mermin}^{(j)} \rangle^2_{1,2, \dots, j} } \right],  
\end{eqnarray}
\end{widetext}
where $\text{Mermin}^{(j)}$ denotes the $j$-party Mermin expression with local bound $2^{(j-1)/2}$ for odd $j$ and $2^{j/2}$ for even $j$. 
At the maximum quantum value of $\langle \text{Mermin}^{(n)} \rangle_{1,\dots, n} = 2^{n-1}$, we get $P_{\text{guess}}(A_1|X_1 = 1, \mathcal{E}) \leq \frac{1}{2}$. 





\textit{Quantum monogamies for a class of inequalities with more than two inputs.-}
Suppose that Alice and Bob receive inputs $x, y \in [m]$ for even $m$ and output $a, b \in \{0,1\}$. We consider the correlation Bell inequality $\mathcal{I}^m$ described by the coefficient matrix $C = (t_{y-x})_{x,y=1}^m$ with
\begin{equation}
\label{eq:gen-ch}
t_l =
\begin{cases}
\phantom{-} 1,
& \text{if } \vert l \vert \leq (m/2) -1 \; \; \vee \; \; l = (m/2),\\
-1, & \text{if } \vert l \vert \geq (m/2)+2 \; \; \vee \; \; l = -(m/2)-1, \\
0 & \text{else}
\end{cases}
\end{equation}
%
As shown in~\cite{LKRG+17}, we have the bound $\beta_c = m^2/2$ for all classical theories, so that
\begin{equation}
\label{eq:cl-val}
\langle \mathcal{I}^n \rangle = \sum_{x,y=1}^n C_{x,y} \langle A_x B_y \rangle \leq m^2/2. 
\end{equation}
This value is achieved when Alice and Bob deterministically set $a, b = 0$ giving $\langle A_x B_y \rangle = 1$ for all $x, y$. 
Eq.(\ref{eq:cl-val}) is proven by writing the coefficient matrix $C$ as a sum of $m^2/4$ CHSH expressions (each with a classical bound of $2$), given as
\begin{eqnarray}
&&\langle A_{j+l-1} \left(B_{j} + B_{j+(m/2)} \right) + A_{j+(m/2) + l-1} \left(B_{j+(m/2)} - B_j \right) \rangle \nonumber \\
&&  \qquad \leq 2 \qquad \qquad \; \; \; \forall j \in [m/2], l \in [m/2]
\end{eqnarray}
The quantum value of the inequality is given by~\cite{LKRG+17}
	\begin{eqnarray}
\beta_q = m \csc \left[ \frac{\pi}{2m} \right], 
\end{eqnarray}
and is achieved when the two parties perform measurements in the equatorial plane on their half of a shared singlet state. Therefore, for large $m$, we have 
\begin{eqnarray}
\frac{\beta_q}{\beta_c} \xrightarrow{\text{m} \rightarrow \infty} \frac{4}{\pi}.
\end{eqnarray} 

\begin{prop}
	\label{prop:mono-chain}
The generalized chain inequality $\mathcal{I}^m$ satisfies
\begin{eqnarray}
\label{eq:ch-mono}
\langle \mathcal{I}^m \rangle^2_{AB} + \langle \mathcal{I}^m \rangle^2_{AC} \leq m^4/2,
\end{eqnarray}
for any quantum state and measurements. In all no-signaling theories, the inequality satisfies the relation
\begin{eqnarray}
\label{eq:ch-mono-ns}
\langle \mathcal{I}^m \rangle_{AB} + \langle \mathcal{I}^m \rangle_{AC} \leq m^2. 
\end{eqnarray}
\end{prop}
\begin{proof}
We know that 
\begin{widetext}
\begin{eqnarray}
\label{eq:chsh-mono-2}
&& 2 \sqrt{2} \mathbf{1} - \cos{\theta} \; \left( A_{j+l-1} \otimes B_j \otimes \mathbf{1} + A_{j+l-1} \otimes B_{j+(m/2)} \otimes \mathbf{1} + A_{j+(m/2) + l-1} \otimes B_{j+(m/2)} \otimes \mathbf{1} - A_{j+(m/2)+l-1} \otimes B_{j} \otimes \mathbf{1} \right) \nonumber \\
&& \qquad - \sin{\theta} \left( A_{j+l-1}\otimes \mathbf{1} \otimes C_j + A_{j+l-1}\otimes \mathbf{1} \otimes C_{j+(m/2)} + A_{j+(m/2) + l-1} \otimes \mathbf{1} \otimes C_{j+(m/2)} - A_{j+(m/2)+l-1} \otimes \mathbf{1} \otimes C_{j} \right) \succeq 0, \nonumber \\
\end{eqnarray} 
\end{widetext}
for $\theta \in [0, \pi/2]$ and $``\succeq"$ denotes the positive semi-definiteness of the operator. This can for instance be seen by a sum-of-squares decomposition of the operator
\begin{widetext}
\begin{eqnarray}
\frac{1}{\sqrt{2}} &&\big[ \left(\mathbf{1} - \cos{\theta} \; \left( \frac{A_{j+l-1} + A_{j+(m/2) + l-1}}{\sqrt{2}} \otimes B_{j+(m/2)} \otimes \mathbf{1} \right) - \sin{\theta} \; \left( \frac{A_{j+l-1} - A_{j+(m/2) + l-1}}{\sqrt{2}} \otimes \mathbf{1} \otimes C_j \right) \right)^2  \nonumber \\
&&+ \left(\mathbf{1} - \cos{\theta} \; \left(\frac{A_{j+l-1} - A_{j+(m/2) + l-1}}{\sqrt{2}} \otimes B_{j} \otimes \mathbf{1} \right) - \sin{\theta} \; \left(\frac{A_{j+l-1} + A_{j+(m/2) + l-1}}{\sqrt{2}} \otimes \mathbf{1} \otimes C_{j+(m/2)} \right) \right)^2 \big] \succeq 0. \nonumber \\
\end{eqnarray}
\end{widetext}
Summing the above expression over all $j, l \in [m/2]$ gives the linear version of the spherical relation (\ref{eq:ch-mono}). Analogously, to derive the no-signaling monogamy relation, we apply the same technique as above using as a building block the CHSH monogamy relation in such theories, namely 
\begin{widetext}
\begin{eqnarray}
\label{eq:chsh-mono-3}
&& 4 -  \langle  A_{j+l-1}  B_j + A_{j+l-1}  B_{j+(m/2)} + A_{j+(m/2) + l-1}  B_{j+(m/2)}   - A_{j+(m/2)+l-1}  B_{j}  \rangle \nonumber \\
&& \qquad - \langle A_{j+l-1}  C_j + A_{j+l-1} C_{j+(m/2)} + A_{j+(m/2) + l-1}  C_{j+(m/2)} - A_{j+(m/2)+l-1} C_{j}  \rangle \geq  0. \nonumber \\
\end{eqnarray} 
\end{widetext}
Summing over $j, l \in [m/2]$ gives the no-signaling trade-off.

\end{proof}
In fact, from the proof above, we see that not only the generalized chain inequality but any Bell expression that can be decomposed into a sum of facet CHSH expressions obeys a hyperspherical monogamy relation within quantum theory. Other inequalities of this type are the well-known class of XOR introduced by Slofstra~\cite{Slofstra10} as games requiring a large amount of entanglement to play optimally. For a graph $G$ with $v$ vertices and $e$ edges, the coefficient matrix $A_G$ for the game is constructed as
having two rows for each edge of $G$, and columns indexed by the vertices. For $(u,v)$ an edge in G with $u < v$, the first row corresponding to $(u,v)$ contains a $1/(4e)$ in the $u$-th column, a $-1/(4e)$ in the $v$-th
column, and zeroes everywhere else. The second row corresponding to $(u,v)$ contains a $1/(4e)$ in both the $u$-th and the $v$-th
column, with zeroes everywhere else. One can directly see that this inequality is a convex combination of multiple CHSH inequalities in which Alice does not know exactly which of these CHSH games she is playing. From this convex decomposition, it follows that a spherical monogamy relation of the type in Prop.~\ref{prop:mono-chain} holds for this class of games. This observation also readily extends to the derivation of trade-off relations for multi-party correlation expressions with more than two inputs per party, when these can be decomposed into a convex combination of facet-defining binary XOR games. 

\textit{Monogamy of Genuine multi-party non-locality.-}
Compared with the scenario of two-party non-locality where $P_{A,B|X,Y}$ is either local or non-local, in the multi-party scenario, different kinds of non-locality can be distinguished. In the tripartite scenario, the fully local correlations $P_{A,B,C|X,Y,Z}$ are those that can be written as $P_{A,B,C|X,Y,Z}(a,b,c|x,y,z) = \sum_{\lambda} q_{\Lambda}(\lambda) P_{A|X,\Lambda}(a|x,\lambda) P_{B|Y, \Lambda}(b|y,\lambda) P_{C|Z,\Lambda}(c|z,\lambda)$. Correlations not in the above form are non-local, however different kinds of non-locality may be distinguished. In seminal work, Svetlichny introduced the notion of genuine $3$-way nonlocal correlations, which are those $P_{A,B,C|X,Y,Z}$ that cannot be written in the form
\begin{eqnarray}
\label{eq:bilocal}
&&P_{A,B,C|X,Y,Z}(a,b,c|x,y,z) = \nonumber \\
&&r_{AB|C} \sum_{\lambda} q_{\Lambda}(\lambda) P_{A,B| X, Y, \Lambda}(a,b|x,y,\lambda) P_{C|Z, \Lambda}(c|z, \lambda) \nonumber + \text{perm.}
\end{eqnarray}
 with $r_{AB|C}, r_{AC|B}, r_{BC|A} \geq 0$, $r_{AB|C} + r_{AC|B} + r_{BC|A} = 1$ and $\sum_{\lambda} q_{\Lambda}(\lambda) = \sum_{\gamma} q_{\Gamma}(\gamma) = \sum_{\upsilon} q_{\Upsilon}(\upsilon) = 1$, where the bipartite correlations can be arbitrary signaling ones. Svetlichny introduced an inequality, the violation of whose bilocal bound guarantees that the correlations are Svetlichny nonlocal, i.e., not of the form in (\ref{eq:bilocal}). The Svetlichny expression for an arbitrary number of parties was shown in~\cite{CGPRS02} in terms of the family of Mermin-Klyshko polynomials $M_n$~\cite{Mermin90, BK93}. 
Letting $M_1 = A^{(1)}_1$, the MK polynomials $M_n$ for $n$ parties are defined recursively by
\begin{eqnarray}
M_n = \frac{1}{2} M_{n-1} \left(A^{(n)}_1 + A^{(n)}_2 \right) + \frac{1}{2} \tilde{M}_{n-1} \left(A^{(n)}_1 - A^{(n)}_2 \right),
\end{eqnarray} 
where $\tilde{M}$ is obtained from $M$ by swapping the observables $1 \leftrightarrow 2$ for all the parties. The Svetlichny polynomial $\mathcal{S}_n$ is defined in terms of $M_n$ by
\begin{eqnarray}
\mathcal{S}_n =
\begin{cases}
M_n  \qquad \qquad \qquad \text{n even } \\    \frac{1}{2} \left(M_n + \tilde{M}_n \right) \quad \text{n odd}
\end{cases}
\end{eqnarray}

We now introduce a generalization of the Svetlichny polynomial to the case of an arbitrary number of inputs. This inequality is a modification of a different generalization with algebraic violation introduced in~\cite{AGCA12, BBBG+12} to show that the GHZ-correlations can be fully genuine multi-partite nonlocal. Accordingly, we let the $i$-th party receive $m$ inputs $x_i \in \{1, \dots, m\}$ for even $m$, whereupon they measure the binary observable $A^{(i)}_{x_i}$ and obtain the outcome $a_i \in \{\pm 1\}$. We then define the Bell expression to be
\begin{eqnarray}
\label{eq:gen-Svetlichny-exp}
\mathcal{S}^{\text{gen}}_{n,m} := \sum_{x_1, \dots, x_n=1}^{m} C_{x_1, \dots, x_n} \langle \prod_{i=1}^{n} A^{(i)}_{x_i} \rangle \leq \beta_c. 
\end{eqnarray} 
The coefficient tensor $\hat{C}$ of our generalized Svetlichny expression $\mathcal{S}^{\text{gen}}_{n,m}$ is constructed recursively starting from
\[
C_{1, \dots, 1, j_n} = \begin{cases}
-1 & \text{if } j_n \leq \frac{m}{2} + 1 \\    1 & \text{if } j_n > \frac{m}{2}+ 1,
\end{cases}
\]
with $C_{1, \dots,1, 1+j_{n-r}, j_{n-r+1},\dots, j_{n-1}, j_n} =$
\begin{eqnarray}
\begin{cases}
C_{1, \dots,1, j_{n-r}, (j_{n-r+1}-1) \; \text{mod m},j_{n-r+2},\dots, j_{n-1}, j_n} & \text{if } j_{n-r+1} > 1  \nonumber \\
-C_{1, \dots,1, j_{n-r}, (j_{n-r+1}-1) \; \text{mod m},j_{n-r+2},\dots, j_{n-1}, j_n} & \text{if } j_{n-r+1} = 1 \nonumber \\
\end{cases}
\end{eqnarray}
for $r = 1, \dots, n-1$. Note that in contrast to the chained Svetlichny expression introduced in~\cite{AGCA12, BBBG+12}, the inequality in (\ref{eq:gen-Svetlichny-exp}) corresponds to a total function, i.e., all $m^n$ possible $n$-tuples of measurement settings $(x_1, \dots, x_n)$ appear in the inequality.  

We claim that the bilocal value for the inequality is 
\begin{eqnarray}
\beta_{bl} = m^{n}/2. 
\end{eqnarray}
This follows by an inductive argument similar to the one considered in~\cite{AGCA12}. For $n=2$, the inequality $\mathcal{S}^{\text{gen}}_{2,m}$ reduces to $\mathcal{I}^m$ from the previous section so that the classical bound $m^2/2$ follows. Suppose the bound $m^{k-1}/2$ holds for $\mathcal{S}^{\text{gen}}_{k-1,m}$. For each setting $x_k$ of the $k$-th party, the inequality $\mathcal{S}^{\text{gen}}_{k,m}(x_k)$ reduces to $\mathcal{S}^{\text{gen}}_{k-1,m}$ up to input-output relabelings so that the bound $m^{k-1}/2$ holds for each of the $m$ settings $x_k$. Summing this bound over all the $m$ settings gives the bilocal value $m^k/2$ for $\mathcal{S}^{\text{gen}}_{k,m}$. The general bound $m^n/2$ then follows by induction.  

We now present a quantum strategy for the $n$ parties and prove that it achieves optimal violation of the inequality. The parties share the $n$-qubit GHZ state $| \text{GHZ}_n \rangle = \frac{1}{\sqrt{2}} \left[ |0 \rangle^{\otimes n} + | 1 \rangle^{\otimes n} \right]$ and measure observables in the $x-y$ plane, i.e., 
\begin{eqnarray}
A^{(i)}_{x_i} = \cos{(\theta^{(i)}_{x_i})} \sigma_x + \sin{(\theta^{(i)}_{x_i})} \sigma_y
\end{eqnarray}
For a given number of parties $n$, we set
\[\theta^{(i)}_{x_i} = 
\begin{cases}
\frac{(-1)^n(x_i-1) \pi}{m} & \text{if } i = 1, 3, \dots, n-2, x_i \in [m] \\
\frac{(-1)^{n+1}(x_i-1) \pi}{m}  & \text{if } i = 2, 4, \dots, n-1,  x_i \in [m] \\
\frac{(2m-(-1)^n-2x_n) \pi}{2m} & \text{if } i = n, x_n \in [m].
\end{cases}
\]
The quantum value achieved by the strategy is given by
\begin{eqnarray}
\beta_q \geq 2m^{n-1} \sum_{i=1}^{m/2} \cos{\frac{(2i-1) \pi}{2m}} = m^{n-1} \csc{\left(\frac{\pi}{2m}\right)}.
\end{eqnarray}
This again closely follows an inductive argument for an analogous quantum strategy in~\cite{AGCA12} showing that the value of $\mathcal{S}_{n,m}^{\text{gen}}$ is equal to $m$ times the value of $\mathcal{S}_{n-1,m}^{\text{gen}}$. The value then follows from the quantum value $(m \csc{\left[\frac{\pi}{2m}\right]})$ of $\mathcal{I}^{m}$ from the previous section. 
We now show that the strategy is optimal. To do this, we show that it achieves the upper bound on $\beta_q$ given as~\cite{Wehner2006, EKB13, MRMC16, RAM16}
\begin{eqnarray}
\beta_q \leq m \sum_{x_3, \dots, x_n=1}^{m} \| C_{*, *, x_3, \dots, x_n} \|
\end{eqnarray}
where $\| \cdot \|$ denotes the spectral norm (maximal singular value). Now, by construction, fixing $x_3, \dots, x_n$ gives the coefficient matrix $C = (t_{y-x})_{x,y=1}^n$ from Eq.(\ref{eq:gen-ch}) up to input-output relabelings. The spectral norm of this matrix was calculated in~\cite{LKRG+17} and shown to be given by (setting $k=m/2$ in the expression in~\cite{LKRG+17}) $\| C_{*, *, x_3, \dots, x_n} \| = m \csc{\left(\frac{\pi}{2m}\right)}$. Summing over $x_3, \dots, x_n$ gives the value $m^{n-1} \csc{\left(\frac{\pi}{2m}\right)}$. The quantum strategy using the $n$-qubit GHZ state achieves this upper bound, so we conclude that it is optimal.

\begin{prop}
Genuine multi-party non-locality evidenced by the generalized Svetlichny polynomial $\mathcal{S}^{\text{gen}}_{n,m}$ for even $n$ is monogamous, i.e.,
\begin{eqnarray}
\langle \mathcal{S}_{n,m}^{\text{gen}} \rangle^2_{P_1} + \langle \mathcal{S}_{n,m}^{\text{gen}} \rangle^2_{P_2} \leq 2 \beta^2_{bl},
\end{eqnarray}
where $P_1$ and $P_2$ are two sets of players with $|P_1| = |P_2| = n$ and $P_1 \cap P_2 \neq \varnothing$.
\end{prop}
\begin{proof}
%
We first show the bound for the case $m=2$, the general bound then follows from the technique in the proof of Prop.~\ref{prop:mono-chain}, i.e., by decomposing the general expression as a sum of the Svetlichny expressions for $m=2$. 
The bilocal bound $\beta_{bl}$ of the $m=2$ Svetlichny expression for any partition $k \mid (n-k)$ of the parties  for even $n$ is given by $2^{\frac{n-2}{2}}$.
The corresponding optimal quantum value is given by $2^{\frac{n-1}{2}}$. 
Being a full-correlation inequality in the $(n,2,2)$ scenario it can be decomposed into the facet inequalities, specifically into the Mermin-Klyshko polynomials, so that the hyperspherical trade-off relations derived in the previous sections holds for this inequality. In particular, from Prop.~\ref{prop:ladder-network} and Prop.~\ref{prop:diff-no-players}, we know that if the two sets intersect at one party $1$, we have
\begin{eqnarray}
\langle \mathcal{S}^{\text{gen}}_{n,2} \rangle^2_{1,2,4,\dots, 2n-2} + \langle \mathcal{S}^{\text{gen}}_{n,2} \rangle^2_{1,3,5,\dots, 2n-1} \leq 2^{n-1},
\end{eqnarray}
which is equivalent to the bound $2 \beta_{bl}^2 (= 2(2^{\frac{n-2}{2}})^2 = 2^{n-1})$. The bound for any arbitrary intersection of the two player sets essentially follows from the above case.
Namely, we have two sets of $2^n$ operators which we would like to group into $2^{n-1}$ sets of four mutually anti-commuting operators each. From the construction in the proof of Prop.~\ref{prop:diff-no-players}, we see that these sets are obtained as $X^{\otimes |P_1 \cap P_2| - k} \otimes Z^{\otimes k} \otimes $ $\{\small{\textsc{XXIXI\dots XI}, \textsc{XXIXI\dots ZI}, \textsc{ZIXIX\dots IX}, \textsc{ZIXIX\dots IZ}} \}$ for fixed $0 \leq k \leq |P_1 \cap P_2|$, i.e., we augment the anti-commuting set in the situation $|P_1 \cap P_2| = 1$ by all combinations of tensor product of $X,Z$ operators at the remaining positions of intersection. 
\end{proof}

\textit{Monogamies within general No-Signaling Theories.-}
In previous sections, we had derived the trade-offs in violation for the $(n,2,2)$ correlation inequalities within quantum theory. In this section, we derive the no-signaling bound on these trade-offs, in the more general $(n,m,d)$ scenario. In doing this, we correct an error in~\cite{PB09} where such general multi-partite no-signaling monogamy relations were derived. 

Consider a general Bell scenario with $n$ parties $A_1, \dots, A_n$. The $i$-th party chooses one among $m_i$ inputs $X_i = x_i \in [m_i]$ and obtains one of $d_i$ outcomes $O_i = o_i \in [d_i]$. Denoting $\textbf{o} := \{o_1, \dots, o_n\}$ and $\textbf{x} := \{x_1, \dots, x_n\}$, the general Bell expression is given as 
\begin{eqnarray}
\label{eq:gen-Bell-ns}
\mathcal{I}_{A_1, \dots, A_n} \equiv \sum_{\textbf{o}, \textbf{x}}  \gamma_{\textbf{o}, \textbf{x}} P(\textbf{o}| \textbf{x}) \leq \beta_c,
\end{eqnarray}
where $\beta_c$ denotes the optimal classical value. General monogamy relations for arbitrary two-party inequalities were derived in~\cite{PB09}. However, the generalization to the multi-party scenario in that paper contained an error which we now rectify in the following proposition. 

\def\d{3} 
\def\a{0.35} 
\def\as{0.27} 
\def\rv{0.05} 
\def\alpha{63.4349489} 
\begin{figure}[t !]
	\label{fig:ns-mono-hypergraph}
	\centering
	\begin{tikzpicture}[thick,scale=0.6, every node/.style={scale=0.6}]
	
	\node (vA) at (0, 0) {};
	\node (vB1) at ({\d}, {\d}) {};
	\node (vB2) at ({\d}, 0) {};
	\node (vB3) at ({\d}, {-\d}) {};
	\node (vC1) at ({2 * \d}, {\d}) {};
	\node (vC2) at ({2 * \d}, 0) {};
	\node (vC3) at ({2 * \d}, {-\d}) {};
	
	\begin{scope}[fill opacity=0.8]
	
					\filldraw[fill=gray!40]
					($(vA) + (0, \a)$) 
					to[out=0,in=180] ($(vC2) + ({0}, {\a})$)
					to[out=0,in=90] ($(vC2) + ({\a}, {0})$)
					to[out=270,in=0] ($(vC2) + ({0}, {-\a})$)
					to[out=180,in=0] ($(vA) + ({0}, {-\a})$)
					to[out=180,in=270] ($(vA) + ({-\a}, {0})$)
					to[out=90,in=180] ($(vA) + ({0}, {\a})$)
					;
			
				\filldraw[fill=gray!35]
					($(vA) + ({-\a / sqrt(2)}, {\a / sqrt(2)})$) 
					to[out=45,in=225] ($(vB1) + ({-\a / sqrt(2)}, {\a / sqrt(2)})$)
					to[out=45,in=180] ($(vB1) + (0, {\a})$)
					to[out=0,in=180] ($(vC1) + (0, {\a})$)
					to[out=0,in=90] ($(vC1) + ({\a}, 0)$)
					to[out=270,in=0] ($(vC1) + (0, {-\a})$)
					to[out=180,in=0] ($(vB1) + ({(sqrt(2)-1) * \a}, {-\a})$)
					to[out=225,in=45] ($(vA) + ({\a / sqrt(2)}, {-\a / sqrt(2)})$)
					to[out=225,in=315] ($(vA) + ({-\a / sqrt(2)}, {-\a / sqrt(2)})$)
					to[out=135,in=225] ($(vA) + ({-\a / sqrt(2)}, {\a / sqrt(2)})$)
					;
				
				\filldraw[fill=gray!30]
					($(vA) + ({-\a / sqrt(2)}, {-\a / sqrt(2)})$) 
					to[out=360-45,in=360-225] ($(vB3) + ({-\a / sqrt(2)}, {-\a / sqrt(2)})$)
					to[out=360-45,in=360-180] ($(vB3) + (0, {-\a})$)
					to[out=360-0,in=360-180] ($(vC3) + (0, {-\a})$)
					to[out=360-0,in=360-90] ($(vC3) + ({\a}, 0)$)
					to[out=360-270,in=360-0] ($(vC3) + (0, {\a})$)
					to[out=360-180,in=360-0] ($(vB3) + ({(sqrt(2)-1) * \a}, {\a})$)
					to[out=360-225,in=360-45] ($(vA) + ({\a / sqrt(2)}, {\a / sqrt(2)})$)
					to[out=360-225,in=360-315] ($(vA) + ({-\a / sqrt(2)}, {\a / sqrt(2)})$)
					to[out=360-135,in=360-225] ($(vA) + ({-\a / sqrt(2)}, {-\a / sqrt(2)})$)
					;
				
				\filldraw[fill=gray!25]
					($(vA) + ({-\a / sqrt(2)}, {\a / sqrt(2)})$) 
					to[out=45,in=225] ($(vB1) + ({-\a / sqrt(2)}, {\a / sqrt(2)})$)
					to[out=45,in=135] ($(vB1) + ({\a / sqrt(2)}, {\a / sqrt(2)})$)
					to[out=315,in=135] ($(vC2) + ({\a / sqrt(2)}, {\a / sqrt(2)})$)
					to[out=315,in=45] ($(vC2) + ({\a / sqrt(2)}, {-\a / sqrt(2)})$)
					to[out=225,in=315] ($(vC2) + ({-\a / sqrt(2)}, {-\a / sqrt(2)})$)
					to[out=135,in=315] ($(vB1) + (0, {-\a * sqrt(2)})$)
					to[out=225,in=45] ($(vA) + ({\a / sqrt(2)}, {-\a / sqrt(2)})$)
					to[out=225,in=315] ($(vA) + ({-\a / sqrt(2)}, {-\a / sqrt(2)})$)
					to[out=135,in=225] ($(vA) + ({-\a / sqrt(2)}, {\a / sqrt(2)})$)
					;
				
				\filldraw[fill=gray!20]
					($(vA) + ({-\a / sqrt(2)}, {-\a / sqrt(2)})$) 
					to[out=360-45,in=360-225] ($(vB3) + ({-\a / sqrt(2)}, {-\a / sqrt(2)})$)
					to[out=360-45,in=360-135] ($(vB3) + ({\a / sqrt(2)}, {-\a / sqrt(2)})$)
					to[out=360-315,in=360-135] ($(vC2) + ({\a / sqrt(2)}, {-\a / sqrt(2)})$)
					to[out=360-315,in=360-45] ($(vC2) + ({\a / sqrt(2)}, {\a / sqrt(2)})$)
					to[out=360-225,in=360-315] ($(vC2) + ({-\a / sqrt(2)}, {\a / sqrt(2)})$)
					to[out=360-135,in=360-315] ($(vB3) + (0, {\a * sqrt(2)})$)
					to[out=360-225,in=360-45] ($(vA) + ({\a / sqrt(2)}, {\a / sqrt(2)})$)
					to[out=360-225,in=360-315] ($(vA) + ({-\a / sqrt(2)}, {\a / sqrt(2)})$)
					to[out=360-135,in=360-225] ($(vA) + ({-\a / sqrt(2)}, {-\a / sqrt(2)})$)
					;
				
				\filldraw[fill=gray!15]
					($(vA) + (0, \as)$) 
					to[out=0,in=180] ($(vB2) + ({(sqrt(2)-1) * \as}, {\as})$)
					to[out=315,in=135] ($(vC3) + ({\as / sqrt(2)}, {\as / sqrt(2)})$)
					to[out=315,in=45] ($(vC3) + ({\as / sqrt(2)}, {-\as / sqrt(2)})$)
					to[out=225,in=315] ($(vC3) + ({-\as / sqrt(2)}, {-\as / sqrt(2)})$)
					to[out=135,in=315] ($(vB2) + ({-(sqrt(2)-1) * \as}, {-\as})$)
					to[out=180,in=0] ($(vA) + ({0}, {-\as})$)
					to[out=180,in=270] ($(vA) + ({-\as}, {0})$)
					to[out=90,in=180] ($(vA) + ({0}, {\as})$)
					;
				
				\filldraw[fill=gray!10]
					($(vA) + (0, -\as)$) 
					to[out=360-0,in=360-180] ($(vB2) + ({(sqrt(2)-1) * \as}, {-\as})$)
					to[out=360-315,in=360-135] ($(vC1) + ({\as / sqrt(2)}, {-\as / sqrt(2)})$)
					to[out=360-315,in=360-45] ($(vC1) + ({\as / sqrt(2)}, {\as / sqrt(2)})$)
					to[out=360-225,in=360-315] ($(vC1) + ({-\as / sqrt(2)}, {\as / sqrt(2)})$)
					to[out=360-135,in=360-315] ($(vB2) + ({-(sqrt(2)-1) * \as}, {\as})$)
					to[out=360-180,in=360-0] ($(vA) + ({0}, {\as})$)
					to[out=360-180,in=360-270] ($(vA) + ({-\as}, {0})$)
					to[out=360-90,in=360-180] ($(vA) + ({0}, {-\as})$)
					;
					
				\filldraw[fill=gray!5]
					($(vA) + ({\a / sqrt(2)}, {\a / sqrt(2)})$) 
					to[out=315,in=135] ($(vB3) + ({-((sqrt(5)-sqrt(2)) / 3) * \as}, {((sqrt(5) + 2*sqrt(2)) / 3) * \as})$)
					to[out=\alpha,in={\alpha+180}] ($(vC1) + ({-(2 / 5) * sqrt(5) * \as}, {(sqrt(5) / 5) * \as})$)
					to[out=\alpha,in={\alpha+90}] ($(vC1) + ({(sqrt(5) / 5) * \as}, {(2 * sqrt(5) / 5) * \as})$)
					to[out={\alpha+270},in=\alpha] ($(vC1) + ({(2 / 5) * sqrt(5) * \as}, {-(sqrt(5) / 5) * \as})$)
					to[out={\alpha+180},in=\alpha] ($(vB3) + ({(sqrt(5) / 3) * \as}, {-(sqrt(5) / 3) * \as})$)
					to[out={\alpha+180},in={315}] ($(vB3) + ({-(sqrt(2) / 3) * \as}, {-(2 * sqrt(2) / 3) * \as})$)
					to[out=135,in=315] ($(vA) + ({-\a / sqrt(2)}, {-\a / sqrt(2)})$)
					to[out=135,in=225] ($(vA) + ({-\a / sqrt(2)}, {\a / sqrt(2)})$)
					to[out=45,in=135] ($(vA) + ({\a / sqrt(2)}, {\a / sqrt(2)})$)
					;
					
				\filldraw[fill=gray!0]
					($(vA) + ({\a / sqrt(2)}, {-\a / sqrt(2)})$) 
					to[out=360-315,in=360-135] ($(vB1) + ({-((sqrt(5)-sqrt(2)) / 3) * \as}, {-((sqrt(5) + 2*sqrt(2)) / 3) * \as})$)
					to[out=360-\alpha,in={360-180-\alpha}] ($(vC3) + ({-(2 / 5) * sqrt(5) * \as}, {-(sqrt(5) / 5) * \as})$)
					to[out=360-\alpha,in={360-90-\alpha}] ($(vC3) + ({(sqrt(5) / 5) * \as}, {-(2 * sqrt(5) / 5) * \as})$)
					to[out={360-270-\alpha},in=360-\alpha] ($(vC3) + ({(2 / 5) * sqrt(5) * \as}, {(sqrt(5) / 5) * \as})$)
					to[out={360-180-\alpha},in=360-\alpha] ($(vB1) + ({(sqrt(5) / 3) * \as}, {(sqrt(5) / 3) * \as})$)
					to[out={360-180-\alpha},in=360-315] ($(vB1) + ({-(sqrt(2) / 3) * \as}, {(2 * sqrt(2) / 3) * \as})$)
					to[out=360-135,in=360-315] ($(vA) + ({-\a / sqrt(2)}, {\a / sqrt(2)})$)
					to[out=360-135,in=360-225] ($(vA) + ({-\a / sqrt(2)}, {-\a / sqrt(2)})$)
					to[out=360-45,in=360-135] ($(vA) + ({\a / sqrt(2)}, {-\a / sqrt(2)})$)
					;
	
	\end{scope}
	
	\fill (vA) circle (\rv) node [below left] {$v_A$};
	\fill (vB1) circle (\rv) node [below left] {$v_{B_1}$};
	\fill (vB2) circle (\rv) node [below right] {$v_{B_2}$};
	\fill (vB3) circle (\rv) node [below right] {$v_{B_3}$};
	\fill (vC1) circle (\rv) node [below left] {$v_{C_1}$};
	\fill (vC2) circle (\rv) node [below right] {$v_{C_2}$};
	\fill (vC3) circle (\rv) node [below right] {$v_{C_3}$};
	
	\end{tikzpicture}
	\caption{Illustration of the hypergraph from Prop.~\ref{prop:genBell-ns-mono} with $3^2 = 9$ hyperedges of the form $\{v_A, v_{B_i}, v_{C_j}\}$, $i,j \in \{1,2,3\}$, $5$ hyperedges shown. Generic trade-off relations hold for the no-signaling values of arbitrary Bell inequalities in such a configuration of players.}
	\label{fig:ns-mono-hypergraph}
\end{figure}
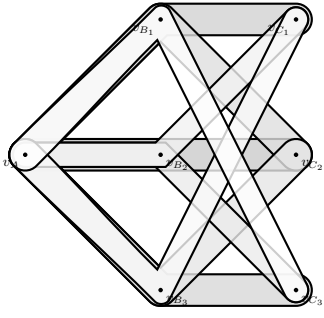

\begin{prop}
	\label{prop:genBell-ns-mono}
	Consider the generalized ladder network of $1 + \sum_{i=2}^{n} m_i$ parties depicted by the hypergraph in Fig.~\ref{fig:ns-mono-hypergraph}. Alice attempts to violate the general Bell expression $\mathcal{I}$ from Eq.(\ref{eq:gen-Bell-ns}) with $m_2$ ``Bobs" $A_2^{(k_2)}$, $m_3$ ``Charlies" $A_3^{(k_3)}$, etc. in $\prod_{i=2}^{n} m_i$ simultaneous Bell experiments. In this scenario, the following trade-off relation holds in all no-signaling theories. 
	\begin{eqnarray}
	\sum_{k_2, \dots, k_n=1}^{m_2, \dots, m_n} \mathcal{I}_{A_1, A_2^{(k_2)}, \dots, A_n^{(k_n)}} \leq \beta_c \prod_{i=2}^n m_i. 
	\end{eqnarray}
\end{prop}

\begin{proof}
	We write 
	\begin{widetext}
	\begin{eqnarray}
	&&\sum_{k_2, \dots, k_n=1}^{m_2, \dots, m_n} \mathcal{I}_{A_1, A_2^{(k_2)}, \dots, A_n^{(k_n)}} \nonumber \\ &&= \sum_{k_2, \dots, k_n=1}^{m_2, \dots, m_n} \sum_{\textbf{o}, \textbf{x}} \gamma_{\textbf{o}, \textbf{x}} P(O_1 = o_1, O^{(k_2)}_2 = o_2, \dots, O^{(k_n)}_n = o_n |  X_1 = x_1, X_2^{(k_2)} = x_2, \dots, X_n^{(k_n)} = x_n) \nonumber \\
	&&=  \sum_{k_2, \dots, k_n = 1}^{m_2, \dots, m_n} \sum_{\textbf{o}, \textbf{x}} \gamma_{\textbf{o}, \textbf{x}} P(O_1 = o_1,  \dots, O_n^{(x_n  + k_n \; \text{mod} \; m_n)}= o_n|X_1 = x_1, X_2^{(x_2+k_2 \; \text{mod} \; m_2)} = x_2, \dots, X_n^{(x_n + k_n \; \text{mod} \; m_n)} = x_n) \nonumber \\
	&&= \sum_{k_2, \dots, k_n=1}^{m_2, \dots, m_n} \mathcal{\tilde{I}}^{k_2, \dots, k_n}_{A_1, A_2, \dots, A_n} \nonumber \\
	&&\leq  \beta_c \prod_{i=2}^n m_i. 
	\end{eqnarray}
	\end{widetext}
	Here $\mathcal{\tilde{I}}^{k_2, \dots, k_n}_{A_1, A_2, \dots, A_n}$ denotes the Bell expression $\mathcal{I}$ written such that for fixed $k_2, \dots, k_n$ each of the parties $A_i^{(l_i)}$ measures a single fixed input $l_i  - k_i \; \text{mod} \; m_i$. As such, a joint probability distribution for the measurement outputs of all the parties can readily be constructed. For instance, for any fixed $k_2, \dots, k_n$, the following joint probability distribution given by
	\begin{widetext}
	\begin{eqnarray}
	&&P(O_1,O_2^{(1)}, \dots, O_2^{(m_2)}, \dots, O_n^{(m_n)}|X_1,X_2^{(1)} = 1-k_2,\dots, X_2^{(m_2)} = m_2 - k_2, \dots, X_n^{(m_n)} = m_n - k_n \; ) := \nonumber \\
	&&\frac{\prod_{l_1 = 1}^{m_1} P(O_1,O_2^{(1)}, \dots, O_2^{(m_2)}, \dots, O_n^{(m_n)}|X_1 = l_1,X_2^{(1)} = 1-k_2,\dots, X_2^{(m_2)} = m_2 - k_2, \dots, X_n^{(m_n)} = m_n - k_n)}{P(O_2^{(1)},\dots, O_2^{(m_2)}, \dots, O_n^{(m_n)}|X_2^{(1)} = 1-k_2,\dots, X_2^{(m_2)} = m_2 - k_2, \dots, X_n^{(m_n)} = m_n - k_n)^{m_1-1}} 
	\end{eqnarray}
	\end{widetext}
	can be directly seen to reproduce all the observable marginal distributions so that a local realistic model exists for the expression $\mathcal{\tilde{I}}^{k_2, \dots, k_n}_{A_1, A_2, \dots, A_n}$. Here, the no-signaling assumption imposes that each of the marginals on the right hand side of the expression is well-defined and independent of the other parties' inputs. 
	%
	Therefore, each $\mathcal{\tilde{I}}^{k_2, \dots, k_n}_{A_1, A_2, \dots, A_n}$ obeys the bound $\beta_c$ in any no-signaling theory, from which the bound on their sum follows. 
	\end{proof}

Applied to the ladder network from Fig.~\ref{fig:ladder}, we see that the above relation exactly gives the hyperplane in Prop.~\ref{prop:ladder-ns} bounding the quantum hyperspherical trade-off relation thus extending the relationship between the quantum and no-signaling trade-off relations for the CHSH inequality found by Toner and Verstraete~\cite{TV06}. While the no-signaling trade-off relation is general, it is an open question whether it is tight, i.e., whether there exists a Bell expression for which a network with $\prod_{i=2}^{n} m_i + 1$ parties is needed before such a monogamy relation manifests itself. For instance, it has been found that for general two-party correlation inequalities, monogamy relations hold in networks with far fewer number of parties~\cite{RH14}.

\textit{Conclusions.-}
In this paper, we have studied the trade-offs in quantum violations of $n$-party full correlation inequalities. Employing an uncertainty relation for complementary binary observables, we derived trade-off relations in several network configurations and showed their tightness by specifying explicit quantum strategies achieving the respective bounds. We then showed that the uncertainty relation does not capture Bell monogamies in their entirety by proving a tight trade-off in a qubit ladder network that does not arise from the uncertainty relation. In deriving a generic trade-off relation between correlation inequalities on different numbers of parties, we discussed how these trade-offs help to characterize a portion of the boundary of the set of quantum correlations. The trade-offs enabled us to show the existence of flat regions in this set, i.e., the existence of Bell inequalities which are optimally violated by multiple distinct quantum strategies (boxes) not related by an isometry. We then studied the trade-offs in violations of a class of Bell inequalities with an arbitrary number of inputs, and showed that genuine multi-party non-locality as evidenced by the generalized Svetlichny polynomial is monogamous. In these cases, our analysis extends to arbitrary full correlation inequalities for any number of parties and inputs, the initial analysis of Toner and Verstraete~\cite{TV06} for the CHSH inequality. Finally, we clarified an error in~\cite{PB09} in the derivation of multi-party monogamy relations based upon no-signaling constraints and compared the hyperspherical quantum trade-off relations with the hypercube relations obtained from general no-signaling constraints alone. 
Important open questions remain. A mathematical question is to characterize the precise configurations of qubits in which tight trade-off relations among quantum correlations exist. A practical task of immediate importance is to apply these trade-off relations in devising device-independent protocols for secure multi-party communications such as secret sharing in the specific network configurations presented here.

\textit{Acknowledgments.-}
We are grateful to Tomasz Paterek and Minh Cong Tran for useful discussions. R.R. also acknowledges useful discussions with Stefano Pironio. This work was supported by the ERC AdG grant QOLAPS. R.R. acknowledges support from the research project  ``Causality in quantum theory: foundations and applications'' of the Fondation Wiener-Anspach and from the Interuniversity Attraction Poles 5 program of the Belgian Science Policy Office under the grant IAP P7-35 photonics@be. P.M. acknowledges support from the National Science Centre (NCN) grant 2014/14/E/ST2/00020 and DS Programs of the Faculty of Electronics, Telecommunications and Informatics, Gda\'nsk University of Technology.

\appendix



\end{document}